%% file: paper.tex
\documentclass[11pt]{article}
\usepackage{natbib}
\usepackage[linesnumbered,ruled,vlined]{algorithm2e}
\usepackage{verbatim}
\usepackage{graphicx}

\usepackage{amsmath}
\usepackage{tikz}
\usetikzlibrary{positioning}
\usepackage{amsmath,amsfonts,amsthm}
\usepackage{thmtools}
\usepackage{xspace}
\usepackage{xcolor}
\usepackage{url}
\usepackage{caption}
\usepackage{subcaption}
\usepackage{tikz}
\usepackage{paralist}
\usepackage{libertine}
\usepackage{dsfont}
\usepackage[
bookmarksnumbered,
unicode,
pagebackref,
colorlinks=true,
urlcolor=cyan,
linkcolor=magenta,
citecolor=blue]
{hyperref}
\usepackage[nameinlink]{cleveref}
\usepackage{enumitem}
\usepackage{todonotes}

\usepackage{thm-restate}

\usepackage[margin=1in]{geometry}
\usepackage{txfonts}
\usepackage[T1]{fontenc}
\usepackage{newtxtext,newtxmath}

\usepackage{silence}
\theoremstyle{definition}

\newtheorem{theorem}{Theorem}[section]
\newtheorem{lemma}[theorem]{Lemma}
\newtheorem{proposition}[theorem]{Proposition}
\newtheorem{corollary}[theorem]{Corollary}

\newcommand{\pfold}{\text{fold}}
\newcommand{\redeps}{$\text{red}_{\text{EPS}}(C)$}
\newcommand{\redcon}{$\text{red}_{\text{concise}}(C)$}
\newcommand{\redps}{$\text{red}_{\text{PS}}(C)$}
\newcommand{\redbgt}{$\text{red}_{\text{BGT}}(A)$}
\newcommand{\redrs}{$\text{red}_{\text{RS}}(A)$}

\newcommand{\trep}{\text{rep}}
\newcommand{\tsma}{\text{small}}
\newcommand{\tmed}{\text{medium}}
\newcommand{\tbig}{\text{big}}
\newcommand{\tsum}{\textsc{clause\_sum}}
\newcommand{\tind}{\textsc{ind}}
\newcommand{\tlev}{\textsc{level}}
\newcommand{\tall}{\textsc{all\_jobs}}
\newcommand{\tjob}{\textsc{jobs}}
\newcommand{\trem}{\textsc{to\_remove}}
\newcommand{\tfac}{\textsc{factor}}
\newcommand{\tnex}{\textsc{next}}
\newcommand{\tper}{\textsc{period}}
\newcommand{\tpers}{\textsc{periods}}

\SetKw{Break}{break}

\title{NP-Hardness and a PTAS for the Pinwheel Problem}
\author{
    Robert Kleinberg\thanks{Cornell University. Emails: {\tt rdk@cs.cornell.edu, abm247@cornell.edu}} \and
    Ahan Mishra\footnotemark[1]
}
\date{}

\begin{document}
\pagenumbering{gobble}
\begin{titlepage}
\maketitle

\input{abstract}

\end{titlepage}

\clearpage

\pagenumbering{arabic}
\clearpage

\input{introduction}

\input{preliminaries}
\input{ptas}
\input{hardness}
\input{variants}

\appendix

\input{ptas-proof}
\input{hardness-proof}

\bibliographystyle{acm}
\bibliography{ref}

\end{document}

%% file: abstract.tex
\begin{abstract}

In the pinwheel problem, one is given an $m$-tuple of positive integers $(a_1, \ldots, a_m)$ and asked whether the integers can be partitioned into $m$ color classes $C_1,\ldots,C_m$ such that every interval of length $a_i$ has non-empty intersection with $C_i$, for $i = 1, 2, \ldots, m$. It was a long-standing open question whether the pinwheel problem is NP-hard. We affirm a prediction of \citet{holte1989pinwheel} by demonstrating, for the first time, NP-hardness of the pinwheel problem. This enables us to prove NP-hardness for a host of other problems considered in the literature: pinwheel covering, bamboo garden trimming, windows scheduling, recurrent scheduling, and the constant gap problem. On the positive side, we develop a PTAS for an approximate version of the pinwheel problem. Previously, the best approximation factor known to be achievable in polynomial time was $\frac{9}{7}$. 

\end{abstract}

%% file: introduction.tex
\section{Introduction}
\label{sec:intro}

In the pinwheel problem, one is given an $m$-tuple of positive numbers called ``periods.'' For a given tuple $(a_1,\ldots,a_m)$, the problem is to decide whether the integers can be partitioned into $m$ color classes $C_1,\ldots,C_m$ such that every interval of length $a_i$ has non-empty intersection with $C_i$, for $i=1,2,\ldots,m$.
This problem models scheduling $m$ repeating tasks each having an upper bound on the idle time between consecutive occurrences: assigning the integer $t$ to color class $C_i$ corresponds to scheduling task $i$ at time $t$, and the periods are interpreted as upper bounds on the amount of idle time allowed between consecutive occurrences of the respective tasks. Owing to this analogy, pinwheel instances having a solution that satisfies the designated constraints are called ``schedulable.''

Studied since the late 1980's, pinwheel problems and their variants have found applications in areas ranging from packet scheduling in communication networks~\citep{holte1989pinwheel}, 
to machine maintenance~\citep{gkasieniec2017bamboo,gkasieniec2024perpetual,wei1983periodic},
to security games~\citep{kempe2018quasi}, 
to spaced repetition of review material in education~\citep{novikoff2012education}.

The computational complexity of solving pinwheel problems is surprisingly ill-understood, given their importance. As stated above, the problem involves searching over an infinite set of colorings of the integers, but it is easily seen to belong to \textsc{PSPACE} by reducing
to the problem of testing for a directed cycle in a
graph (of size exponential in the input description) 
whose adjacency relation is computable in small space \citep{holte1989pinwheel}. The pinwheel problem is not known to belong to any complexity class below \textsc{PSPACE}.

It was a long-standing open question whether the pinwheel problem is NP-hard.  
\citet{korst1996scheduling} showed 
that the pinwheel problem with \emph{exact
periods} --- i.e., the variant that 
requires consecutive occurrences of task $i$ 
to be spaced exactly $a_i$ apart --- is 
NP-hard. This hardness result was shown 
by \citet{bar2007windows} 
to imply the NP-hardness of a compactly 
encoded pinwheel problem in which periods
have multiplicities, and the multiplicities 
may be encoded in binary. A variant of the pinwheel problem 
in which jobs are to be scheduled only finitely many times 
is also known to be NP-hard \citep{finite}. 

For the pinwheel problem itself, an unpublished 
manuscript by \citet{jacobs2014new}
presents a proof that the problem has no
pseudo-polynomial time algorithm unless
SAT is solvable in quasi-polynomial time
by a randomized algorithm. A work by \cite{kobayashihardness} shows that the pinwheel problem has no polynomial-time algorithm unless SAT is solvable in quasi-polynomial-time by a deterministic algorithm. However, before this paper, it was not known whether a polynomial-time algorithm for the pinwheel problem could be eliminated based only on the assumption that $P \neq NP$.

In light of their hardness, one may wonder 
about the approximability of pinwheel problems. 
To make the question meaningful, one can use a reformulation of the pinwheel problem, due to \citet{kawamura}, which allows for fractional periods. The (potentially fractional) period vector $(a_1,\ldots,a_m)$ is considered schedulable if there exists a partition of the integers into color classes $C_1,C_2,\ldots,C_m$ such that in each finite subinterval of length $L$, the number of elements of color $C_i$ is at least $\lfloor L/a_i \rfloor$, for $i=1,2,\ldots,m$. (Note that when $(a_1,\ldots,a_m) \in \mathbb{N}^m$, the two definitions of schedulability coincide.) For a given approximation factor $\alpha > 1$, one may now consider the $\alpha$-approximate pinwheel problem, 
i.e.~the problem that requires deciding whether an integer $m$-tuple $(a_1,\ldots,a_m)$ 
is not schedulable or the tuple $(\alpha \cdot a_1, \ldots, \alpha \cdot a_m)$ is schedulable. 

A similar problem has been studied in the context of the Bamboo Garden Trimming (BGT) problem. In particular, in the BGT problem, one is given an $m$-tuple of positive numbers called ``growth rates''. For a given tuple $(h_1, \dots, h_m)$, the problem is to produce a partition of the integers into $m$ color classes $C_1, \dots, C_m$ which minimizes $\max_{i \in [1, m]} h_i D_i$ where $D_i$ is the maximum distance between any two consecutive occurrences of task $C_i$. Several prior works on state-of-the-art BGT approximation guarantees rely on producing an efficient algorithm that decides whether a pinwheel instance $(a_1, \dots, a_m)$ is unschedulable or $(\lfloor \alpha \cdot a_1 \rfloor, \dots, \lfloor \alpha \cdot a_m \rfloor)$ is schedulable. Specifically, solving this problem for any particular $\alpha$ implies an $\alpha$-approximation for BGT. Prior works have achieved approximation factors of $2$ \citep{gkasieniec2017bamboo}, $\frac{32000}{16947}$ \citep{bgt188}, $\frac{12}{7}$ \citep{bgt127}, $1.6 + \epsilon$ \citep{gkasieniec2024perpetual}, $\frac{10}{7}$ \citep{bgt107}, $\frac{4}{3}$ \citep{kawamura}, and $\frac{9}{7}$ \citep{mishra}. Due to the monotonicity property of pinwheel scheduling \citep{kawamura} with $\lfloor \alpha \cdot a_i \rfloor \leq \alpha \cdot a_i$, an efficient algorithm for the version of the problem with floors implies one for our problem. Prior to our work,
no polynomial-time approximation algorithm was 
known for any value of $\alpha$ smaller than~$\frac{9}{7}$. 

\subsection{Our results and techniques}

In this paper, we take two steps toward resolving the computational complexity 
of the pinwheel problem. First, we prove that it admits a PTAS: for any constant
$\epsilon > 0$, the $(1+\epsilon)$-approximate pinwheel problem can be solved in 
polynomial time. In fact, the algorithm we present is an EPTAS, meaning that its
running time on inputs of size $m$ is $O_{\epsilon}(m^c)$ where 
$c$ is a universal constant independent of $\epsilon$, while the
$O_{\epsilon}(\cdot)$ masks a constant that depends on $\epsilon$.
Second, we prove that the pinwheel problem is NP-Hard.
Our hardness result is based on a reduction that always
produces pinwheel problem instances whose  ``density'' 
$D(A) = \sum_{i=1}^m \frac{1}{a_i}$ equals 1. 
As a consequence of this property, the same reduction
can also be used to show the NP-hardness of many related
problems: pinwheel covering \citep{covering}, 
bamboo garden trimming \citep{gkasieniec2017bamboo},
windows scheduling \citep{windowsoriginal}, 
recurrent scheduling \citep{recharging}, and the 
constant gap problem \citep{kubiak2004fair}. 

\paragraph{Technical overview of PTAS:}
Our PTAS is based on partitioning the tasks into three types --- big, medium, 
and small --- based on their periods. Tasks with shorter periods need to be 
scheduled more frequently and are accordingly classified as big or medium. The
cutoffs are chosen so that the total number of big and medium tasks is bounded 
by a constant depending only on the approximation parameter $\epsilon$. 
The algorithm uses a brute-force iteration over all $p$-periodic schedules
of big tasks and gaps (i.e., idle time slots), using a simple density-based 
criterion to decide whether the schedule contains enough gaps to accommodate the
medium and small tasks. 

A similar partial-enumeration idea was used by 
\citet{recharging} in their work on the Recharging Bandits problem, but 
in that problem, the schedule is unconstrained: the cost of violating the
scheduling constraint for a task is merely a reduction in the value of the
objective function. The PTAS of \citeauthor{recharging}
is thus able to omit the medium tasks from the schedule 
altogether, enabling a ``separation of scales'' that facilitates
the analysis. The difference between scheduling a small task near
the beginning of one cycle of the $p$-periodic schedule and scheduling
it near the end of that cycle is negligible compared to the spacing 
between consecutive occurrences of the task. This means that one 
only needs to focus on including enough gaps per cycle, not on 
how the gaps are distributed within the cycle. 

Designing a PTAS for the pinwheel problem is more difficult because 
the (relaxed) scheduling constraints for each task are treated as hard 
constraints, so the medium tasks cannot be omitted from the schedule 
and the separation of scales cannot be achieved. Instead, our PTAS succeeds
in finding a valid schedule for the relaxed problem by enforcing an
upper bound on the spacing between consecutive gaps in the $p$-periodic 
schedule and scheduling the medium tasks into these gaps 
using the ``fold'' operation of \citet{kawamura}. Finally, to show
that the small tasks can be scheduled into the remaining gaps, we 
use our algorithm's density-based criterion in combination 
with a theorem of \citet{covering} regarding the schedulability of 
pinwheel problem instances whose periods are all large. 

\paragraph{Technical overview of NP-hardness:}
To prove NP-hardness of pinwheel scheduling, we enhance the 
existing reduction from 3-SAT to \emph{exact pinwheel scheduling} (EPS)
due to \citet{exact}. A key observation is that a pinwheel 
scheduling instance $A$ with density $D(A)=1$ is schedulable
if and only if it remains schedulable when regarded as an 
instance of EPS. Thus, our strategy is to modify the reduction
of \citet{exact} by ``padding'' its output with additional tasks
that increase the density to 1, while preserving the reduction's 
completeness and soundness. Carrying out this strategy is delicate
because preserving completeness is surprisingly subtle. 

A useful metaphor is to visualize the reduction of \citet{exact} as 
yielding a set of jigsaw puzzle pieces that are guaranteed 
to fit disjointly inside a rectangular frame when the reduction
is applied to a satisfiable 3-SAT instance. 
(The analogy is only metaphorical, since the EPS problem involves
packing arithmetic progressions into the set of integers, not 
packing two-dimensional geometric shapes into a rectangle.)
The original reduction yields a pinwheel scheduling instance 
whose density is strictly less than 1, akin to leaving gaps 
between the puzzle pieces when they are arranged inside the frame. 
The challenge of ``padding'' the  reduction to yield density 1 is 
then akin to manufacturing a set of additional puzzle
pieces that are guaranteed to fill the remaining gaps, no matter
how the original pieces are positioned within the frame. 
An easy strategy would be to manufacture a large quantity of 
identical grains of sand tiny enough to fill in the gaps regardless
of the gaps' sizes and shapes. Analogously, one strategy for
padding the EPS reduction to yield density 1 is to use 
exponentially many identical tasks, each with an exponentially 
large period. Indeed, this strategy succeeds in proving  
the NP-hardness of a compact encoding of pinwheel scheduling
in which the multiplicities of identical tasks are encoded in binary;
see \Cref{lem:concise} and \cite{bar2007windows}.
 
To prove NP-hardness of pinwheel scheduling, without the
compact encoding, we need to find a method of padding the 
reduction using only polynomially many additional tasks. 
Our approach begins with specializing the source problem of
the reduction from 3-SAT to 3,4-SAT --- the special case of 3-SAT
in which every variable appears in at most four clauses, which
is still NP-complete. The bounded-multiplicity structure of
the 3,4-SAT instance enables us to find a polynomially-bounded
set of ``allowed job periods'' and to show that a collection of
tasks with these periods can always be arranged to fill in the
gaps in the EPS schedule arising when the reduction is applied
to a satisfiable 3,4-SAT instance.

\label{sec:results}

%% file: preliminaries.tex
\section{Preliminaries} \label{sec:prelim}

An \emph{instance} of the pinwheel scheduling problem is specified by a set of tasks, $M$, 
and a period length, $a_i \ge 1$, for each $i \in M$. Unless otherwise specified,
we standardize on the task set $M = [m]$, in which case a pinwheel instance is 
represented by a vector of periods, $A = (a_1,\ldots,a_m)$.

A \emph{schedule} for an instance with task set $M$ is a function 
$\sigma : \mathbb{Z} \to M \cup \{ \bot \}$, where $\sigma(t)=i$ is interpreted 
to mean that task $i$ is scheduled at time $t$, and $\sigma(t) = \bot$ is
interpreted to mean that no task is scheduled at time $t$; in the latter
case time step $t$ is called a \emph{holiday} in the schedule. 
For a positive integer $p$, a schedule $\sigma$ is called $p$-periodic if it 
satisfies $\sigma(t) = \sigma(t+p)$ for all $t \in \mathbb{Z}$. 
Schedule $\sigma$
is called \emph{valid} for pinwheel scheduling instance $A$ if, for each 
interval $I$ of $L$ consecutive integers and each task $i \in M$, the set 
$\sigma^{-1}(\{i\}) \cap I$ has at least $\lfloor L / a_i \rfloor$ elements.
The instance is called \emph{schedulable} if it has a valid schedule,
and otherwise it is called \emph{unschedulable}. 

A key property of pinwheel scheduling instances is their \emph{density},
$D(A) = \sum_{i=1}^m \frac{1}{a_i}$. When this density exceeds 1, the
instance is unschedulable. To see why, consider any interval $I$ made up 
of $L$ consecutive integers and observe that any valid schedule $\sigma$ must satisfy
\[
  \sigma^{-1}([m]) \cap I \geq
  \sum_{i=1}^m \left\lfloor \frac{L}{a_i} \right\rfloor >
    \sum_{i=1}^m \left( \frac{L}{a_i} - 1 \right) = L \left( D(A) - \frac{m}{L} \right) .
\]
If $D(A) > 1$ and $L > \frac{m}{D(A)-1}$ then the right side exceeds $L$, 
whereas the left side cannot exceed $L$ since $I$ has only $L$ elements.

On the positive side, \cite{kawamura} proved that every pinwheel instance $A$ with $D(A) \geq \frac{5}{6}$ is schedulable. A key technique in proof of \cite{kawamura} is the following fold operation, related to the folding idea from \cite{towards}.

\begin{algorithm}[ht]
  \DontPrintSemicolon
  \caption{Fold Operation, $\pfold_{\theta}(A)$.}\label{algo:fold}
  \KwIn{Pinwheel instance $A = (a_1, a_2, \dots, a_m)$}
  
  \While{$\max(A) > \theta$}{
    Let $a$ and $b$ be the largest and second largest values in $A$, allowing for repetition.
    
    \If{$b > \theta$}{    
    $A \gets A \ominus (a, b)$

    $A \gets A \sqcup (b/2)$
    }
    \Else{
    $A \gets A \ominus (a)$

    $A \gets A \sqcup (\theta)$
    }
    \Return $A$
}
\end{algorithm}

The notation $A \sqcup (a)$ represents adding a job of period $a$ to instance $A$ and $A \ominus (a)$ represents removing a job of period $a$ from instance $A$. We will also use the notation $[n]$ to denote the set $\{ 1, \dots, n \}$. 

\begin{lemma}[\citealp{kawamura}, Lemma 4] \label{lem:fold}
  For any pinwheel scheduling instance $A$ and any $\theta > 0$,
  \begin{enumerate}
      \item  If $\pfold_{\theta}(A)$ is schedulable, so is $A$.
      \item  Any task in $\pfold_{\theta}(A)$ with period 
      $\leq \frac{\theta}{2}$ is already in $A$.
      \item $D(\pfold_{\theta}(A)) < D(A) + \frac{1}{\theta}$.
  \end{enumerate}
\end{lemma}

To any pinwheel scheduling instance $A = (a_1,a_2,\ldots,a_m)$ with integer
periods, one may associate a \emph{state graph} $G(A)$, which is a 
directed graph with vertex set $[a_1] \times [a_2] \times \cdots \times [a_m].$
The state graph has a directed edge from $(j_1,j_2,\ldots,j_m)$ 
to $(j'_1, j'_2, \ldots, j'_m)$ if $j'_i \in \{1,j_i+1\}$ for all $i$,
and at most one of $j'_1,\ldots,j'_m$ is equal to 1. 
Valid schedules $\sigma$ are in bijection
with doubly-infinite walks in this state graph. In one 
direction, the bijection associates to any valid 
schedule $\sigma$ the walk with vertex sequence 
$\{j(t) = (j_1(t),\ldots,j_m(t))\}_{t \in \mathbb{Z}}$
where $j_i(t)$ is the minimum $j \geq 1$ such that $\sigma(t-j) = i$.
The inverse bijection maps a walk $\{j(t)\}$ to the schedule
$\sigma$ defined by setting $\sigma(t)$ equal to the unique 
index $i$ such that $j_i(t)=1$, or $\sigma(t) = \bot$ if there
is no such $i$. Since a directed graph has a doubly-infinite 
directed walk if and only if it contains a directed cycle, we 
obtain the following lemma of \citet{holte1989pinwheel}.

\begin{lemma}[\citealp{holte1989pinwheel}, Theorem 2.1]
    \label{lem:periodic}
    A pinwheel scheduling instance $A = (a_1,\ldots,a_m)$ with 
    integer periods has a valid schedule if and only if it 
    has a $p$-periodic valid schedule for some $p \leq \prod_{i=1}^m a_i$.
\end{lemma}

%% file: ptas.tex
\section{PTAS for Pinwheel Scheduling}\label{sec:ptas}

For an approximation factor $\alpha>1$, we define the following 
\emph{$\alpha$-approximate decision problem} for pinwheel scheduling.
Given an integer $m$-tuple $A = (a_1,\ldots,a_m)$, the decision procedure must do one of the following.
\begin{enumerate}
    \item Correctly declare that the instance is unschedulable, or
    \item Correctly declare that the instance $\alpha \cdot A = (\alpha \cdot a_1,\ldots,\alpha \cdot a_m)$ 
    is schedulable. 
\end{enumerate}
Note that if both options 1 and 2 are possible, the decision procedure may choose either option. 

This section is devoted to presenting a polynomial-time approximation scheme for pinwheel scheduling. In
other words, we will show that for each constant $\epsilon > 0$, there is a polynomial-time 
algorithm for the $(1+\epsilon)$-approximate pinwheel scheduling problem. In fact, 
\Cref{algo:pptas} below is  
an \emph{efficient polynomial-time approximation scheme} (EPTAS), satisfying a running time
bound of $O \left( 2 \uparrow \uparrow O(\frac{1}{\epsilon})\, \cdot \, m \right), $
where we have used Knuth's up-arrow notation \citep{knuth} for the tetration, or ``tower'' function. Furthermore,
in the satisfiable case (case 2 above, when it declares that $(1+\epsilon) \cdot A$ is schedulable)
the algorithm can be enhanced to 
output an efficient representation of a valid schedule for $(1+\epsilon) \cdot A$.

\begin{algorithm}[t]
\caption{Polynomial-time approximation scheme for pinwheel scheduling}
\label{algo:pptas}
\KwIn{Pinwheel scheduling instance $A = (a_1, \dots, a_m)$ and parameter $\epsilon \in (0, 2/7)$}

$n \gets \left\lceil \frac{1}{\epsilon} \right\rceil$

\If{$D(A) > 1$}{
\Return{``$A$ \textnormal{is unschedulable}''}
}

$\ell \gets n$

$u \gets 16n^2 \ell^{\ell}$

\While{\textnormal{True}}{

Let $A_{\ell < j < u}$ be defined as the list of job periods in $A$ that are between $\ell$ and $u$.

\If{$D(A_{\ell < j < u}) \leq \frac{1}{2(n + 1)}$}{
\Break
}

$\ell \gets u$

$u \gets 16n^2 \ell^{\ell}$
}

$A_{\tbig} \gets A_{j \leq \ell}$, the list of job periods in $A$ that are at most $\ell$

$p \gets 1$

$h_{\max} \leftarrow 0$

\While{$p \leq \ell^{\ell}$}{
    \For{\textnormal{all valid} $p$\textnormal{-periodic schedules} $s$ \textnormal{of} $A_{\textnormal{big}}$}{
        Let $h$ be the fraction of holidays in one period of $s$
    
        $h_{\max} \gets \max(h_{\max}, h)$
    }

    $p \gets p + 1$
}

$A_{\text{not big}} \gets A_{j > \ell}$, the list of job periods in $A$ that are greater than $\ell$

\If{$h_{\textnormal{max}} < D(A_{\textnormal{not big}})$}{
\Return{``$A$ \textnormal{is unschedulable}''}
}
\Else{
\Return{``$A(1 + \epsilon)$ \textnormal{is schedulable}''}
}
\end{algorithm}

We sketch the analysis of \Cref{algo:pptas} in this 
section and defer the full details of the proof to 
\Cref{sec:ptas-proof} due to space constraints.

\begin{restatable}{lemma}{lemmaruntime}
\label{lem:runtime}
\Cref{algo:pptas} has a running time of $O(m)$, where the big-$O$ hides multiplicative and additive dependencies involving $\lceil 1/\epsilon \rceil$. 
\end{restatable}

For the while loop on Line 6 of \Cref{algo:pptas}, the key idea is that it will run for at most $2(n + 1)$ iterations due to the relevant $A_{\ell < j < u}$ job subsets being disjoint with a total density of at most 1. Then, the $\ell$ and $u$ values can be upper bounded by a (very large) constant (with respect to $m$), and this makes the while loop on Line 15 of \Cref{algo:pptas} essentially constant time. 

\begin{restatable}{lemma}{lemmaunschedulable}
\label{lem:unschedulable}
If \Cref{algo:pptas} outputs ``$A$ is unschedulable'' then $A$ is unschedulable.
\end{restatable}

This lemma would follow if we could show that in every period of a periodic schedule of $A_{\tbig}$, the fraction of holidays is at most $h_{\max}$. Then, with $D(A_{\text{not big}})$ strictly greater than $h_{\max}$ there cannot be space for the jobs that are not big in any periodic schedule (which is guaranteed to exist if $A$ is schedulable by \Cref{lem:periodic} \citep{holte1989pinwheel}) even if they are scheduled perfectly evenly. 

To do this, we show that the holiday fraction for a period of length at most $\ell^{\ell}$ which forms a periodic schedule is an upper bound for the holiday fraction of a period of any periodic schedule. In particular, we show that a period of any periodic schedule with period greater than $\ell^{\ell}$ can be decomposed into two separate periods of smaller corresponding periodic schedules so that the holiday fraction of the original larger period is a convex combination of that of the smaller ones.  

\begin{restatable}{lemma}{lemmaschedulable}
\label{lem:schedulable}
If \Cref{algo:pptas} outputs ``$A(1 + \epsilon)$ is schedulable'' then $A(1 + \epsilon)$ is schedulable.
\end{restatable}

We divide the jobs that are not big into two categories: the medium jobs (which have periods strictly between the final values of $\ell$ and $u$) and the small jobs (which have periods of at least $u$).

Suppose $S$ is the schedule for the big jobs corresponding to $h_{\max}$. If there were no medium jobs, then it would be possible to adapt the techniques of \cite{recharging} to schedule the big jobs within the gaps in $S$. The main technical challenge we overcome is providing a mechanism to schedule the medium jobs. 

In particular, we show that holidays can be inserted every $1/\epsilon$ steps within $S$, and, due to the upper bound on the density of the medium jobs guaranteed by Line 8 of \Cref{algo:pptas}, these inserted holidays are enough to fully schedule the medium jobs. An important component of the proof is the use of the fold operation, and as a subsidiary result, we show that every real pinwheel instance with density at most $1/2$ is schedulable, taking a first step towards a conjecture of \cite{realps}. With some casework, a careful balancing of parameters, and a result from \cite{covering}, we are able to schedule the small jobs as well. The proof is constructive in that it provides a polynomial-time method to produce an efficient representation \citep{holte1989pinwheel} of the final schedule of $A$. 

Considering the proof of efficiency with respect to $m$ (\Cref{lem:runtime}) and the proof of correctness (\Cref{lem:unschedulable} and \Cref{lem:schedulable}), we have justified the following theorem.  

\begin{theorem}
\Cref{algo:pptas} is a PTAS for pinwheel scheduling. 
\end{theorem}

%% file: hardness.tex
\section{NP-Hardness of Pinwheel Scheduling}\label{sec:hardness}

In this section, we present a reduction establishing that the pinwheel scheduling problem is NP-Hard. 
We sketch the analysis of the reduction in this 
section and defer the full details of all proofs to 
\Cref{sec:hardness-proof} due to space constraints.

\citet{exact} proved the corresponding hardness claim for the \emph{exact pinwheel scheduling} (EPS)
problem, in which a schedule is only considered valid if the spacing between every two
consecutive occurrences of task $i$ is \emph{exactly} $a_i$. 
Our reduction carefully adds jobs to the existing reduction to reach density 1 in a way that preserves completeness of the reduction. Soundness follows easily from the fact that jobs are only added, not removed. We start the proof with a simple result about EPS schedules. 

\begin{lemma}[\citealp{exact}, Lemma 2]
\label{lem:epsperiodic}
Let $S$ be an EPS schedule for $A = (a_1, \dots, a_m)$. Then, $S$ is periodic with period $L = \text{LCM}(a_1, \dots, a_m)$. 
\end{lemma}

\begin{proof}
For any $i \in [m]$, suppose job $i$ is scheduled at position $p$. Then, by the definition of EPS schedule, it is scheduled at positions $p + a_i k$ for all $k \in \mathbb{N}$. By the definition of LCM, $k_0 = L/a_i$ is an integer. Then, job $i$ is scheduled at every position of the form $p + a_i k_0 k = p + Lk$ for all $k \in \mathbb{N}$, as desired. 
\end{proof}

We continue with a relevant lemma that appears not to have been stated in the literature before. 

\begin{restatable}{lemma}{lemmaepsfill}
\label{lem:epsfill}
Suppose there is an EPS schedule for an instance $A = (a_1, \dots, a_m)$. Let $L = \text{LCM}(a_1, \dots, a_m)$, and suppose we have $B = (b_1, \dots, b_n)$ where $L|b_i$ for all $1 \leq i \leq n$ and $b_i | b_j$ for all $1 \leq i \leq j \leq n$. Let $C = A \sqcup B$ and suppose that $D(C) \leq 1$. Then, there is an EPS schedule for $C$. 
\end{restatable}

The idea behind the proof is that the jobs in $B$ can be scheduled within the gaps left within an EPS schedule for $A$, and we show that scheduling within these gaps can be considered as a separate instance with rescaled periods that is schedulable due to a theorem of \cite{holte1989pinwheel}. 

The key property of the EPS scheduling problem that permits constructing
``gadgets'' encoding combinatorial constraints is that any two distinct tasks 
$i$ and $j$ must be slotted into distinct congruence classes modulo
$\gcd(a_i,a_j)$. 

\begin{lemma} \label{lem:gcd}
  In any valid schedule for an EPS instance $A = (a_1,\ldots,a_m)$,
  for any two distinct tasks $i$ and $j$, the time steps allocated
  to $i$ and the time steps allocated to $j$ must belong to distinct
  congruence classes modulo $\gcd(a_i,a_j)$.
\end{lemma}
\begin{proof}
  Let $d = \gcd(a_i,a_j)$ and let $q = \frac{a_i}{d}, \, r = \frac{a_j}{d}$. 
  Suppose $x_i$ is the congruence class modulo $a_i$ used for scheduling
  task $i$ and $x_j$ is the congruence class modulo $a_j$ used for scheduling
  task $j$. Write $x_i = y_i d + z_i$, where $0 \leq z_i < d$, and
  write $x_j = y_j d + z_j$, where $0 \leq z_j < d$. Note that all
  occurrences of task $i$ (resp., $j$) are in congruence class 
  $z_i$ (resp., $z_j$) modulo $d$, so the lemma's conclusion can
  be restated as asserting that $z_i \neq z_j$. To prove this
  inequation,  note that $q$ and $r$ are relatively prime, so 
  by the Chinese Remainder Theorem, there is a natural number $y$ such
  that $y \equiv y_i \pmod{q}$ and $y \equiv y_j \pmod{r}$. 
  Then, time step $y d + z_i$ is allocated to task $i$ 
  while $y d + z_j$ is allocated to task $j$, so $z_i \neq z_j$
  as claimed.
\end{proof}

To establish a base for our reduction, we  present an exposition 
of the theorem of \citet{exact} on NP-hardness of exact pinwheel scheduling.

\begin{proposition}[\citealp{exact}, Theorem 1]
Exact Pinwheel Scheduling (EPS) is NP-hard. 
\end{proposition}

\begin{proof}
The reduction due to \citet{exact} has the following form.

\begin{algorithm}[H]
\caption{Polynomial-time reduction from 3-SAT to EPS, \redeps}\label{algo:epsreduction}

\KwIn{3-SAT instance $C = \bigwedge_{j=1}^m C_j,\quad C_j = \bigvee_{k=1}^{3} c_{jk}$, with variables $x_1, \dots, x_n$}

\For{$i = 1, \dots, 2n$}{
$p_{i} \gets$ the $i^{\text{th}}$ prime greater than $\max(m, n)$
}

$J \gets ()$

\For{$i = 1, \dots n$}{
$\trep_1(x_i) \gets p_{2i - 1}$

$\trep_1(\bar{x}_i) \gets p_{2i}$
}

\For{$i = 1, \dots, n$}{
$\trep_2(x_i) \gets 2n[\trep_1(x_i)]^2$

$b_{x_i} \gets [\trep_1(x_i)]^2 - \trep_1(x_i)$

\For{$j = 1, \dots, b_{x_i}$}{
$J \gets J \sqcup (\trep_2(x_i))$
}
}

\For{$i = 1, \dots, n$}{
$\trep_2(\bar{x}_i) \gets 2n[\trep_1(\bar{x}_i)]^2$

$b_{x_i} \gets [\trep_1(\bar{x}_i)]^2 - \trep_1(\bar{x}_i)$

\For{$j = 1, \dots, b_{\bar{x}_i}$}{
$J \gets J \sqcup (\trep_2(\bar{x}_i))$
}
}

\For{$i = 1, \dots, n$}{
$f_i \gets 2n \, \trep_1(x_i) \, \trep_1(\bar{x}_i)$

$J \gets J \sqcup (f_i)$
}

\For{$i = 1, \dots, j$}{
$\trep_2(C_j) \gets 2n[\trep_1(c_{j 1}) \, \trep_1(c_{j 2}) \, \trep_1(c_{j 3})]^2$

$J \gets J \sqcup (\trep_2(C_j))$
}

\Return{$J$}

\end{algorithm}
The proof that the reduction of \Cref{algo:epsreduction} is sound and complete was said in \cite{exact} to be within an in-preparation full version of the paper. However, that version appears not to have been published, so we provide a proof of correctness of this reduction, for the sake of completeness. Note that we will assume, as \cite{exact} do, that each literal is distinct, and a clause cannot contain both a variable and its negation. Clearly, these assumptions do not affect the NP-hardness of 3-SAT. 

\begin{restatable}{lemma}{lemepsredsat}
\label{lem:epsredsat}
If \redeps \text{} has an EPS schedule, then $C$ is satisfiable.
\end{restatable}

We use the same main ideas to prove that $C$ is satisfiable as those mentioned by \cite{exact}, though we present them in greater detail. A relevant technique in the proof is the notion of partitioning an EPS schedule $S$ into subschedules which are essentially congruence classes modulo $2n$. Because each job period is a multiple of $2n$, a job must occupy a single subschedule, and it is possible to associate each literal with its own subschedule, because the $\trep_1(x_i)$ values are relatively prime. The $f_i$ jobs should be interpreted as taking up all the remaining space (after scheduling the $\trep_2(x_i)$ jobs) within each subschedule they are in and preventing any clause jobs from being placed within that subschedule. This is associated with a false literal in the 3-SAT formula. 

\begin{restatable}{lemma}{lemepssatred}
\label{lem:epssatred}
If $C$ is satisfiable, then \redeps \text{} has an EPS schedule. 
\end{restatable}

We use the same notion of a subschedule from the proof of \Cref{lem:epsredsat}. We show that there is an analog of the monotonicity property of pinwheel scheduling \citep{kawamura} for exact pinwheel scheduling. Using this to make all periods within a subschedule the same (or making one period a factor of the rest of the equal periods), we are able to leverage \Cref{lem:epsfill} to prove the existence of an EPS schedule for \redeps.

It can be seen that \Cref{algo:epsreduction} runs in polynomial time, so since 3-SAT is NP-hard \citep{karp}, we have proven NP-hardness for EPS. 
\end{proof}

As noted in \cite{jacobs2014new}, there appears to be a bug in the proof for the NP-hardness of pinwheel scheduling in concise form that is given in \cite{bar2007windows}. We provide an alternative proof, which also serves as a warm-up to our desired result. 

\begin{proposition} \label{lem:concise}
Pinwheel scheduling is NP-hard for a concise representation in which the number of jobs of a given period is encoded in binary. 
\end{proposition}

\begin{proof}
Consider the following reduction.

\begin{algorithm}[H]
\caption{Polynomial-time reduction from 3-SAT to Concise PS, \redcon}\label{algo:concise}

\KwIn{3-SAT instance $C = \bigwedge_{j=1}^m C_j,\quad C_j = \bigvee_{k=1}^{3} c_{jk}$, with variables $x_1, \dots, x_n$}

$J \gets$ \redeps

$L \gets \text{LCM}(J)$, the least common multiple of all job periods in $J$  

$J_c \gets$ a concise representation of $J$ which has keys corresponding to job periods and values corresponding to multiplicity. 

\If{\textnormal{key $L$ is not in $J_c$}}{
$J_c[L] \gets 0$
}

$m \gets L \cdot (1 - D(J))$

$J_c[L] \gets J_c[L] + m$

\Return{$J_c$}

\end{algorithm}

We claim that \Cref{algo:concise} is a complete and sound reduction. 

\begin{lemma}\label{lem:redconschedulable}
If \redcon \text{} is schedulable, then $C$ is satisfiable.
\end{lemma}

\begin{proof}
Since \redcon \text{} is schedulable, there is some corresponding schedule, $S$. Now, note that 
$$D(\text{red}_{\text{concise}}(C)) = D(\text{red}_{\text{EPS}}(C)) + \frac{L \cdot (1 - D(\text{red}_{\text{EPS}}(C)))}{L} = D(\text{red}_{\text{EPS}}(C)) + 1 - D(\text{red}_{\text{EPS}}(C)) = 1.$$
In other words, \redcon \text{} is a dense instance, meaning that $S$ is an EPS schedule for \redcon \text{} \citep{covering}. Now, if we delete the jobs that were added with period $L$ and their associated schedulings, we are left with an EPS schedule for \redeps, since those jobs are unaffected. Then, by \Cref{lem:epsredsat}, $C$ is satisfiable. 
\end{proof}

\begin{lemma}\label{lem:redconsatisfiable}
If $C$ is satisfiable, then \redcon \text{} is schedulable.
\end{lemma}

\begin{proof}
Let us define $A = $ \redeps, and let $B$ be the pinwheel instance corresponding to the added jobs (so $L \cdot (1 - D(\text{red}_{\text{EPS}}(C)))$ repetitions of a job with period $L$). Note that since all periods in $B$ are the same, each period is divisible by the last. We also know that \redcon \text{} is dense (see proof of \Cref{lem:redconschedulable}), so $D(A \sqcup B) \leq 1$. By \Cref{lem:epssatred}, there is an EPS schedule for $A$. Additionally, recalling the definition of $L$ as the LCM of the periods in $A$, \Cref{lem:epsfill} applies. In particular, there is an EPS schedule (and therefore a pinwheel schedule) for \redcon. 
\end{proof}

Observe that \Cref{algo:concise} runs in polynomial time, which completes the proof of NP-hardness. 
\end{proof}

Note that the polynomial running time of \Cref{algo:concise} is contingent on the concise representation because for \redeps, the least common multiple $L$ is exponentially large (despite having polynomially many bits) and $1 - D(J)$ is lower bounded by $1/n^c$ for some $c$, meaning that $m = L \cdot (1 - D(J))$ is exponentially large. In standard representation, adding this many jobs of period $L$ would take exponential time, which would break the reduction. 

In order to improve this result to show NP-hardness of pinwheel scheduling in the standard representation, one natural direction is to apply a similar idea as \Cref{algo:concise} of filling up the instance with the LCM, but instead filling up each subschedule (one for each literal) separately. Recall that a subschedule is a congruence class of a schedule modulo $2n$. One immediate hurdle is that even within a subschedule, the LCM of the jobs that may need to be scheduled may be exponentially large, due to each variable possibly being in polynomially many clauses. The following lemma shows that this is not the inherent difficulty of the problem. 

\begin{lemma}[\citealp{tovey}, Theorem 2.3]
3,4-SAT is NP-hard. 
\end{lemma}

The main idea behind this lemma is to replace a variable that appears $k$ times with $k$ separate variables embedded in a cyclic structure to ensure consistent truth values. We refer to the work of \cite{tovey} for the details. 

For the remainder of this section, when we use the notation $\trep_1(x_i)$, $\trep_2(x_i)$, $f_i$, or $\trep_2(C_j)$, it refers to the corresponding values set in \Cref{algo:epsreduction}. 

Since 3,4-SAT is a special case of 3-SAT, the correctness of \redeps \text{} still holds, and it can be used as a foundation for other reductions that extend \redeps \text{} by adding further jobs (as in \Cref{algo:concise}). Let us define $C(x_i)$ as the set of clauses that contain variable $x_i$. The LCM of the jobs that may possibly go into a subschedule of variable $x_i$ is a factor of 
\begin{equation}
\label{eq:lcm}
2n \, \trep_1(x_i)^2 \, \trep_i(\bar{x}_i)^2 \prod_{C_j \in C(x_i)} \prod_{k = 1}^3 \, (\trep_1(c_{jk}))^2.
\end{equation}
For 3,4-SAT this number is bounded by a polynomial function of $n$ because $x_i$ is involved in at most 4 clauses which each have 3 literals, so ignoring $2n$, there are at most $2 + 2 + 2 \cdot 3 \cdot 4 = 28$ primes we must multiply together, and the following lemma shows that each prime is polynomial in the input size (which is at least $\max(m, n)$). 

\begin{restatable}{lemma}{lemprimesize}
\label{lem:primesize}
Let $v = \max(m, n)$. For all $i \in [n]$, $v < \trep_1(x_i) \leq v^3$ and $v < \trep_1(\bar{x}_i) \leq v^3$. 
\end{restatable}

The left side of the inequalities follows immediately from the process on Line 2 of \Cref{algo:epsreduction}. The right sides follow from a fact about the distribution of prime numbers, which we prove using a result of \cite{counting}.

The next difficulty for the reduction is harder to surpass. In particular, we do not know which literals will be set to true while constructing the pinwheel reduction, so ``filling up'' a subschedule naively will either leave no space for clause jobs which might have been needed, or will leave space for clause jobs which might not have been needed, destroying the completeness of the reduction in either case. 

To solve this, we utilize some more subtle techniques which effectively allow jobs of the appropriate period length to be filled in respective subschedules \textit{after} the truth values of the literals are known. 

We start by introducing the notion of a \textit{variable subschedule}. In particular, it is the union of the subschedules of a literal and its negation, so the total density of a variable subschedule is $\frac{1}{2n} + \frac{1}{2n} = \frac{1}{n}$. We denote the subschedule of variable $x_i$ as $S_i$ (as opposed to $S_{x_i}$ for the subschedule of literal $x_i$). 

We will add jobs corresponding to each variable subschedule so that the remaining density within that subschedule is exactly the sum of densities for all $m$ clause jobs, which we denote as 
$$\tsum := \sum_{j = 1}^m \frac{1}{\trep_2(C_j)}.$$
This can be intuitively thought of as allowing space for all the $m$ clause jobs, though having sufficient density and being able to schedule the jobs are not equivalent. One nuance in this step is that because we are reducing from 3,4-SAT, we could leave space for only the respective clause jobs --- which number at most 4 --- but we do not do this because it wouldn't work with the second step of the reduction. 

Due to leaving space for all $m$ clause jobs, it no longer works to simply add jobs of period length equal to the LCM (the product of the 28 relevant primes), since the density fractions of the other clause jobs have denominators relatively prime to our LCM so it would not be possible to fill up variable subschedules exactly (with total density $\frac{1}{n}$). We develop the following list of allowed job periods to enable an exact density match while also preventing exponentially large numbers of jobs. 

\begin{algorithm}[H]
\caption{Allowed Job Periods, $\text{allowed\_periods}(\tind, C)$}\label{algo:allowed}

\KwIn{Variable index $\tind \in [n]$, 3,4-SAT instance $C = \bigwedge_{j=1}^m C_j,\quad C_j = \bigvee_{k=1}^{3} c_{jk}$, with variables $x_1, \dots, x_n$}

$\tpers \gets []$

$\tper \gets 2n$

\For{$i = 0, 1, \dots, n - 1$}{
$\textsc{clauses} \gets 0$

$\tnex \gets (\tind + i - 1 \bmod{n}) + 1$

\For{\textnormal{each clause $C_j$ that contains the variable $x_{\tnex}$}}{
\For{$k = 1, 2, 3$}{
$\tper \gets \tper \cdot \trep_1(c_{jk})^2$
}

$\textsc{clauses} \gets \textsc{clauses} + 1$
}

$\tper \gets \tper \cdot \trep_1(x_{\tnex})^{2} \cdot \trep_1(\bar{x}_{\tnex})^{2} \cdot \trep_1(x_{\tnex})^{(4 - \textsc{clauses}) \cdot 6}$

Append $\tper$ to $\tpers$
}

\Return{$\textnormal{\tpers}$}

\end{algorithm}

Let us define $\tpers_i := \text{allowed\_periods}(i, C)$ for all $i \in [n]$.  We will interpret $\tpers_i$ as being 0-indexed. Note that $\tpers_i[0]$ should be interpreted as being a multiple of the LCM from \Cref{eq:lcm}. We now state relevant lemmas about \Cref{algo:allowed}. 

\begin{lemma}
\label{lem:shared}
For any $i, j \in [n]$, $\tpers_i[n - 1] = \tpers_j[n - 1]$. 
\end{lemma}

\begin{proof}
For any $i \in [n]$, the set of values attained by $\tnex$ (from Line 5 of \Cref{algo:allowed}) is exactly $\{ 0, 1, \dots, n - 1 \}$ with no repetition. This is because \Cref{algo:allowed} wraps around to all possible values. Multiplication is commutative, so the final product at the end is the same for any $i \in [n]$. 
\end{proof}

Let us define $P$ as the shared value of $\tpers[n - 1]$ guaranteed by \Cref{lem:shared}. 

\begin{lemma}
\label{lem:allowedcyclic}
For any $i \in [n]$ and $j \in [n - 1]$, 
$$2n \, \tpers_i[j] = \tpers_i[0] \, \tpers_{(i \text{ mod } n) + 1}[j - 1].$$
\end{lemma}

\begin{proof}
The result is equivalent to proving that
$$\frac{\tpers_i[j]}{2n} = \frac{\tpers_i[0]}{2n} \cdot \frac{\tpers_{(i \text{ mod } n) + 1}[j - 1]}{2n}$$

The quantity $\tpers_i[j]/2n$ can be produced by multiplying all the primes from $j + 1$ iterations of the loop on Line 3 of \Cref{algo:allowed}. Note that the $2n$ in the denominator comes from the initial setting of $\tper = 2n$ from Line 2 of \Cref{algo:allowed}. 

We separate these $j + 1$ iterations into the first iteration and the remaining $j$ iterations. The first iteration produces primes which multiply together to $\frac{\tpers_i[0]}{2n}$, by definition of $\tpers_i[0]$. The remaining $j$ iterations, due to the cyclic nature of the $\tnex$ variable on Line 5 of \Cref{algo:allowed}, multiply together to $\tpers_{(i \text{ mod } n) + 1}[j - 1]/2n$. Then, 
$$\frac{\tpers_i[j]}{2n} = \frac{\tpers_i[0]}{2n} \cdot \frac{\tpers_{(i \text{ mod } n) + 1}[j - 1]}{2n},$$
as desired. 
\end{proof}

\begin{restatable}{lemma}{lemperiodsize}
\label{lem:periodsize}
For any $i \in [n]$, we have $2n v^{28} \leq \tpers_i[0] \leq 2n v^{84}$. For any $i \in [n]$ and $j \in [n - 1]$,  
$$\text{} v^{28} \leq \frac{\tpers_i[j]}{\tpers_i[j - 1]} \leq v^{84}$$
\end{restatable}

We show that in each iteration of the for loop on Line 3 of \Cref{algo:allowed}, $\tper$ is multiplied by exactly 28 primes, as in \Cref{eq:lcm}, and we use \Cref{lem:primesize} to get upper and lower bounds on the size of each prime. 

Note that while \Cref{algo:allowed} produces periods with exponential value (which also means that our proof establishes weak NP-hardness rather than strong NP-hardness), the entire algorithm runs in polynomial time. 

As discussed earlier, it will be easy to show that our reduction is sound. For completeness, the idea is, for all $i \in [n]$, to assign the variable subschedule $S_i$, a list of jobs that have periods contained within $\tpers_i$. Then, the added jobs can be viewed as the $B$ instance within \Cref{lem:epsfill} (since $\tpers_i[j]$ clearly divides $\tpers_i[k]$ for any $0 \leq j \leq k \leq n - 1$) and the existing jobs (e.g., $\trep_1(x_i)$ and $f_i$ jobs) can be viewed as the $A$ instance, so \Cref{lem:epsfill} can be applied to show schedulability. The first step of the reduction explicitly assigns each subschedule a certain density of jobs, which all have periods within $\tpers_i$, using the following algorithm. 

\begin{algorithm}[H]
\caption{Greedy Job Addition Algorithm, $\text{greedy}(\tind, d, C)$}\label{algo:greedy}

\KwIn{Variable index $\tind \in [n]$, $d \in \left[ 0, \frac{1}{n} \right]$ such that $dP \in \mathbb{Z}$, 3,4-SAT instance $C = \bigwedge_{j=1}^m C_j,\quad C_j = \bigvee_{k=1}^{3} c_{jk}$, with variables $x_1, \dots, x_n$}

$\tpers \gets \textsc{allowed\_periods}(\tind, C)$

$\tall \gets ()$

\For{$i = 1, \dots, n$}{
$\tper \gets \tpers[i - 1]$

$\tjob \gets \lfloor \tper \cdot d \rfloor$

$\begin{aligned}
d \gets d - \frac{\tjob}{\tper}
\end{aligned}$

\For{$l = 1, \dots, \textnormal{\tjob}$}{
Add $\tper$ to $\tall$
}
}

\Return{$\textnormal{\tall}$}

\end{algorithm}

\begin{figure}[htbp]
    \centering
    \begin{minipage}{0.4\textwidth}
        \centering
        \includegraphics[width=\linewidth]{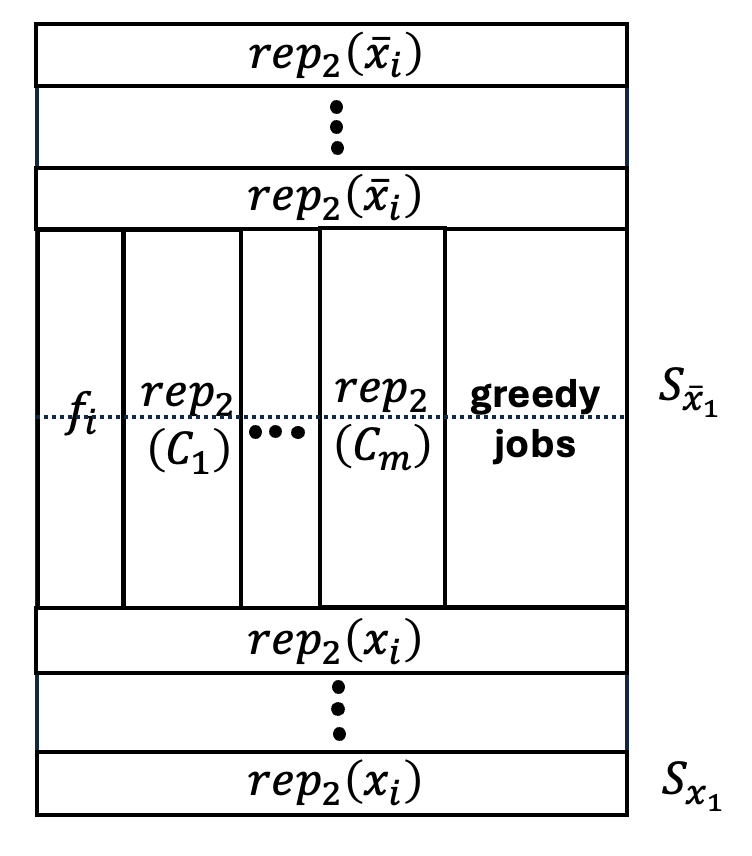}
        \caption{Greedy job additions fill up variable subschedules.}
        \label{fig:greedy}
    \end{minipage}
    \hfill
    \begin{minipage}{0.5\textwidth}
        \centering
        \includegraphics[width=\linewidth]{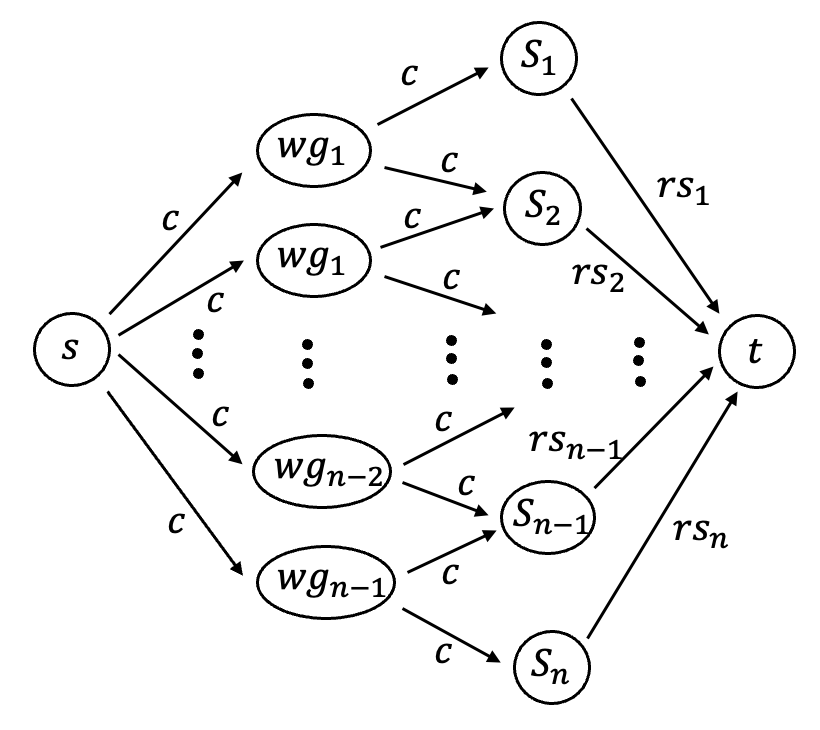}
        \caption{Density Flow Graph. Note that $c = \tsum$ and $rs_i$ is the remaining space in the subschedule of variable $x_i$}
        \label{fig:flow}
    \end{minipage}
\end{figure}

Note that this algorithm will be used for each variable to leave density within each subschedule that is exactly enough for all the clauses to be scheduled within it. \Cref{fig:greedy} provides a visual of the idea behind the first step of the reduction. In particular, the greedy jobs have exactly enough density to fill up the variable subschedule $S_i$ if all of the clause jobs are also placed within $S_i$. Note that crossing the lines between $S_{x_1}$ and $S_{\bar{x}_1}$ does not indicate that the corresponding job can simultaneously occupy both literal subschedules, but rather than it could occupy either, and in total both subschedules would be filled if all displayed jobs from \Cref{fig:greedy} were scheduled within $S_{x_1}$ or $S_{\bar{x}_1}$. In addition, \Cref{fig:greedy} is not to scale. We now state relevant lemmas about \Cref{algo:greedy}.

\begin{restatable}{lemma}{lemgreedypol}
\label{lem:greedypol}
\Cref{algo:greedy} runs in polynomial time. 
\end{restatable}

We show that on each iteration of Line 3 of \Cref{algo:greedy}, the value of $\tjob$ is polynomial in the input size by leveraging \Cref{lem:primesize} and the structure of \Cref{algo:allowed}. 

\begin{restatable}{lemma}{lemgreedyd}
\label{lem:greedyd}
$D(\text{greedy}(\tind, d, C)) = d$. 
\end{restatable}

We show that \Cref{algo:greedy} preserves the invariant that the value of $d$ in each iteration of Line 3 is an integer multiple of $\frac{1}{P}$, so $dP$ is an integer on the last iteration. 

We now move on to the second step of the reduction. 

\Cref{algo:warm} will be called for all $\tind \in [n - 1]$. Then, as shown in \Cref{fig:warm}, \Cref{algo:warm} will add, for all $i \in [n - 1]$, jobs that can be either scheduled within subschedule $S_i$ or subschedule $S_{i + 1}$. In particular, recall that the goal is to assign each subschedule $S_j$, jobs with periods in $\tpers_j$. It is clear how to do this for $j = i$ since \Cref{algo:warm} produces jobs with periods in $\tpers_j$ when called with $\tind = j$. 

For $j = i + 1$, although the jobs generated do not have periods within $\tpers_{j}$, we will leverage the cyclic guarantee of \Cref{lem:allowedcyclic} to show that the jobs constructed by \Cref{algo:warm} can be ``converted'' to jobs which have periods that are in $\tpers_j$ (see the proof of \Cref{lem:psredsat} for details). The warm start (Lines 3 to 7) and skipping $i = 1$ in Line 8 of \Cref{algo:warm} enable the conversion of elements from $\tpers_i$ to elements of $\tpers_{i + 1}$. 

\begin{algorithm}[H]
\caption{Warm Greedy Job Addition Algorithm, $\text{warm\_greedy}(\tind, d, C)$}\label{algo:warm}

\KwIn{Variable index $\tind \in [n]$, $\begin{aligned}d \in \left[ \frac{\tpers[0]}{n \, \tpers[1]}, \frac{1}{n} \right] \end{aligned}$ (see Line 1) such that $dP \in \mathbb{Z}$, 3,4-SAT instance $C = \bigwedge_{j=1}^m C_j,\quad C_j = \bigvee_{k=1}^{3} c_{jk}$, with variables $x_1, \dots, x_n$}

$\tpers \gets \text{allowed\_periods}(\tind, C)$

$\tall \gets ()$

\For{$i = 3, 4, \dots, n$}{
$\begin{aligned}
\tjob \gets \frac{\tpers[0] \, \tpers[i - 1]}{2n \, \tpers[i - 2]}
\end{aligned}$

\vspace{5pt}

$\begin{aligned}
d \gets d - \frac{\tjob}{\tpers[i - 1]}
\end{aligned}$

\For{$l = 1, \dots, \textnormal{\tjob}$}{
Add $\tpers[i - 1]$ to $\tall$
}
}

\For{$i = 2, 3, \dots, n$}{
$\tjob \gets \left\lfloor \tpers[i - 1] \cdot d \right\rfloor$

$\begin{aligned}
d \gets d - \frac{\tjob}{\tpers[i - 1]}
\end{aligned}$

\For{$l = 1, \dots, \textnormal{\tjob}$}{
Add $\tpers[i - 1]$ to $\tall$
}
}

\Return{$\textnormal{\tall}$}

\end{algorithm}

\begin{figure}[htbp]
    \centering
    \includegraphics[width=0.8\textwidth]{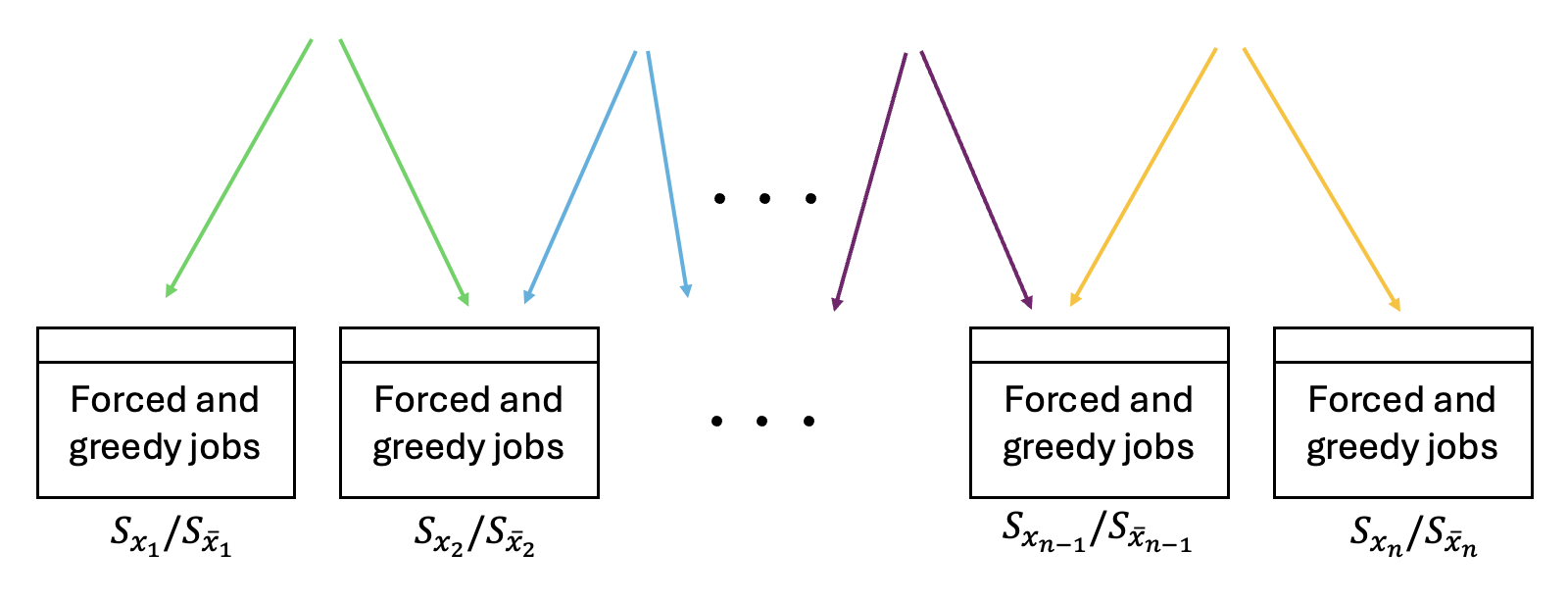}
    \caption{Warm Greedy Job Addition. An amount of density equal to $\tsum$ is left within each variable subschedule after the forced jobs and greedy jobs, and the goal of the warm greedy jobs is to fill any remaining space that is not needed by clause jobs.}
    \label{fig:warm}
\end{figure}

We now state the analogs of \Cref{lem:greedypol} and \Cref{lem:greedyd} for \Cref{algo:warm}.

\begin{restatable}{lemma}{lemwarmpol}
\label{lem:warmpol}
\Cref{algo:warm} runs in polynomial time.
\end{restatable}

The proof is similar to that of \Cref{lem:greedypol}. 

\begin{restatable}{lemma}{lemwarmd}
\label{lem:warmd}
$D(\text{warm}(\tind, d, C)) = d$. 
\end{restatable}

The proof is similar to that of \Cref{lem:greedyd}, with the main complication being that the for loop on Line 3 of \Cref{algo:warm} could add enough jobs that $d$ becomes negative, and we show that this cannot happen as long as $d \geq \frac{\tpers[0]}{n \, \tpers[1]}$, which is the precondition. 

We are now ready to present the full reduction from 3,4-SAT to pinwheel scheduling. 

\begin{algorithm}[H]
\caption{Polynomial-time reduction from 3,4-SAT to pinwheel scheduling, \redps}\label{algo:reduction}

\KwIn{3,4-SAT instance $C = \bigwedge_{j=1}^m C_j,\quad C_j = \bigvee_{k=1}^{3} c_{jk}$, with variables $x_1, \dots, x_n$}

$J \gets$ \redeps

$\begin{aligned}
\tsum \gets \sum_{j = 1}^m \frac{1}{\trep_2(C_j)}
\end{aligned}$

\For{$i = 1, \dots, n$}{
$\begin{aligned}
d_{\text{current}} \gets \frac{b_{x_i}}{\trep_2(x_i)} + \frac{b_{\overline{x}_i}}{\trep_2(\overline{x}_i)} + \frac{1}{f_i} + \tsum
\end{aligned}$

\vspace{3pt}

$\begin{aligned}
d_{\text{remaining}} \gets \frac{1}{n} - d_{\text{current}}
\end{aligned}$

$\tjob \gets \text{greedy}(i, d_{\text{remaining}}, C)$

$J \gets J \sqcup \tjob$
}

\For{$i = 1, \dots, n$}{
$\tjob \gets \text{warm\_greedy}(i, \tsum, C)$

$J \gets J \sqcup \tjob$
}

\Return{$J$}

\end{algorithm}

We justify that \Cref{algo:reduction} demonstrates NP-hardness of the pinwheel problem. 

\begin{lemma}\label{lem:redpol}
\Cref{algo:reduction} runs in polynomial time and satisfies the preconditions of \Cref{algo:greedy} and \Cref{algo:warm}.
\end{lemma}

\begin{proof}
Note that $\tsum$ is composed of a sum of fractions, each of which is an integer multiple of $1/P = 1/\tpers[n - 1]$ (due to Line 8 of \Cref{algo:allowed}), so $\tsum$ itself is an integer multiple of $1/P$. Similarly, $d_{\text{remaining}}$ is composed of a sum and difference of fractions, each of which is an integer multiple of $1/P$, so $d_{\text{remaining}}$ is an integer multiple of $1/P$ as well, satisfying the precondition of \Cref{algo:greedy}. 

We prove that the \Cref{algo:warm} requirement of $d \geq \frac{\tpers[0]}{n \, \tpers[1]}$ holds by proving that $\tsum \geq \frac{2}{\tpers_i[0]}$ for any $i \in [n]$. This is sufficient because by \Cref{lem:allowedcyclic}, 
$$\frac{\tpers_i[0]}{n \, \tpers_i[1]} = \frac{2}{\tpers_{(i \text{ mod } n) + 1}[0]}.$$
Now, for any variable $x_i$, there must be some clause containing $x_i$ (either the positive or negative literal), and by \Cref{algo:epsreduction}, it has period 
$$2n \, (\trep_1(x_i) \, \trep_1(x_a) \, \trep_1(x_b))^2,$$
where $x_a$ and $x_b$ are other literals. By \Cref{lem:primesize}, 
$$2n \, (\trep_1(x_i) \, \trep_1(x_a) \, \trep_1(x_b))^2 \leq 2n \, (\max(m, n)^3 \cdot \max(m, n)^3 \cdot \max(m, n)^3)^2 \leq 2 \max(m, n)^{19}.$$
Since every element in $\tsum$ is positive, we have $\tsum \geq 1/(2 \max(m, n)^{19})$. By \Cref{lem:periodsize}, $\tpers_i[0] \geq \max(m, n)^{28}$, and since $\max(m, n) \geq 3$ (in 3,4-SAT), we have that
$$\frac{1}{2 \max(m, n)^{19}} = \frac{\max(m, n)^9}{4} \cdot \frac{2}{\max(m, n)^{28}} \geq \frac{2}{\max(m, n)^{28}},$$
so $\tsum \geq 2/\tpers_i[0]$, as desired. 

Since \Cref{algo:reduction} satisfies the preconditions of its subroutines, \Cref{lem:greedypol} and \Cref{lem:warmpol} show that each line of \Cref{algo:reduction} can be executed in polynomial time. In addition, the for loops on Lines 3 and 8 of \Cref{algo:reduction} run a polynomial number of times, so the entirety of \Cref{algo:reduction} runs in polynomial time. 
\end{proof}

\begin{lemma}\label{lem:density}
$D(\text{red}_{\text{PS}}(C)) = 1$. 
\end{lemma}

\begin{proof}
The total density after Line 1 of \Cref{algo:reduction} is 
$$\tsum + \sum_{i = 1}^n \left( \frac{b_{x_i}}{\trep_2(x_i)} + \frac{b_{\overline{x}_i}}{\trep_2(\overline{x}_i)} + \frac{1}{f_i} \right).$$
From \Cref{lem:greedyd}, we know that the density added collectively by all iterations of the while loop on Line 3 of \Cref{algo:reduction} is 
$$\sum_{i = 1}^n \left[ \frac{1}{n} - \left( \frac{b_{x_i}}{\trep_2(x_i)} + \frac{b_{\overline{x}_i}}{\trep_2(\overline{x}_i)} + \frac{1}{f_i} + \tsum \right) \right].$$ 
By \Cref{lem:warmd}, we know that the density added collectively by all iterations of the while loop on Line 8 of \Cref{algo:reduction} is $\sum_{i = 1}^{n - 1} \tsum$. Adding everything together, the total density at the end is
$$\tsum + \sum_{i = 1}^n \left( \frac{b_{x_i}}{\trep_2(x_i)} + \frac{b_{\overline{x}_i}}{\trep_2(\overline{x}_i)} + \frac{1}{f_i} \right) + \sum_{i = 1}^n \left[ \frac{1}{n} - \left( \frac{b_{x_i}}{\trep_2(x_i)} + \frac{b_{\overline{x}_i}}{\trep_2(\overline{x}_i)} + \frac{1}{f_i} + \tsum \right) \right]$$
$$+ \sum_{i = 1}^{n - 1} \tsum = \sum_{i = 1}^n \frac{1}{n} + \sum_{i = 1}^n \tsum - \sum_{i = 1}^n \tsum + \sum_{i = 1}^n \left( \frac{b_{x_i}}{\trep_2(x_i)} + \frac{b_{\overline{x}_i}}{\trep_2(\overline{x}_i)} + \frac{1}{f_i} \right)$$
$$- \sum_{i = 1}^n \left( \frac{b_{x_i}}{\trep_2(x_i)} + \frac{b_{\overline{x}_i}}{\trep_2(\overline{x}_i)} + \frac{1}{f_i} \right) = \frac{n}{n} + 0 + 0 = 1,$$

as desired. 
\end{proof}

We now claim that \Cref{algo:reduction} is a sound and complete reduction. 

\begin{restatable}{lemma}{lempsredsched}
\label{lem:psredsched}
If \redps \text{} is schedulable, then $C$ is satisfiable. 
\end{restatable}

The proof is similar to that of \Cref{lem:redconschedulable}. 

\begin{restatable}{lemma}{lempsredsat}
\label{lem:psredsat}
If $C$ is satisfiable, then \redps \text{} is schedulable. 
\end{restatable}

We might hope that the proof of this lemma is as simple as that of \Cref{lem:redconsatisfiable}. While the end of the proof is similar, due to the more complicated nature of our reduction, the entirety of the proof is much more involved and requires the analysis of several auxiliary algorithms (\Cref{algo:flow}, \Cref{algo:combine}, and \Cref{algo:splitter}). 
These algorithms and their analysis are presented in \Cref{sec:hardness-proof}.
At a high level, we show that there is a flow saturating the edges leading out of $s$ in \Cref{fig:flow} (\Cref{algo:flow}), and provide a method to decide which jobs to assign to each subschedule out of two options in \Cref{algo:splitter}. A relevant subroutine to \Cref{algo:splitter} is \Cref{algo:combine}, which is used to reverse the warm start to \Cref{algo:warm}, as needed. Finally, we show that the jobs assigned to each variable subschedule can be scheduled along with the existing jobs using \Cref{lem:epsfill}. 

Considering the proof that \Cref{algo:reduction} runs in polynomial time (\Cref{lem:redpol}) and the proof of correctness (\Cref{lem:psredsched} and \Cref{lem:psredsat}), we have justified the following theorem.

\begin{theorem}\label{thm:pinwheel}
Pinwheel Scheduling is NP-hard.
\end{theorem}

%% file: variants.tex
\section{NP-Hardness of Related Problems}\label{sec:related}

In this section, we prove NP-hardness (and in some cases NP-completeness) of several problems treated previously in the literature. The proofs are based on \Cref{lem:density} and \Cref{thm:pinwheel}. For the motivations and applications of each problem, we refer to the original work introducing the problem. 

Consider the problem of dense pinwheel scheduling \citep{holte1989pinwheel}, where one must decide whether a pinwheel instance $(a_1, \dots, a_n)$ is schedulable, provided that $\sum_{i = 1}^n \frac{1}{a_i} = 1$.

\begin{corollary}\label{cor:dps}
Dense Pinwheel Scheduling is NP-complete.
\end{corollary}

\begin{proof}
By \Cref{lem:density},  \Cref{algo:reduction} only reduces to dense instances. Along with the fact that \Cref{algo:reduction} is a valid reduction (\Cref{thm:pinwheel}), this establishes NP-hardness of dense pinwheel scheduling. In addition, as noted in \cite{holte1989pinwheel}, dense pinwheel scheduling is in NP due to the relevant modular congruence relations being checkable in polynomial time. Together, this proves NP-completeness.
\end{proof}

Consider the pinwheel covering problem \citep{kawamura2020, covering}, where one is given an instance $(a_1, \dots, a_m)$ and the problem is to decide whether the integers can be partitioned into $m$ color classes $C_1,\ldots,C_m$ such that every interval of length $a_i$ has intersects with $C_i$ at most once, for $i=1,2,\ldots,m$. 

\begin{corollary}
Pinwheel Covering is NP-hard.
\end{corollary}

\begin{proof}
As noted in \cite{covering}, for instances with density 1, pinwheel scheduling and pinwheel covering are the same problem, so the NP-hardness of dense pinwheel scheduling (\Cref{cor:dps}) implies NP-hardness of pinwheel covering. 
\end{proof}

Consider the decision version of the bamboo garden trimming problem \citep{gkasieniec2017bamboo} in which one is given an instance $(h_1, \dots, h_m)$ along with a parameter $K$. The problem is to decide whether there is a partition of the integers into $m$ color classes $C_1, \dots, C_m$ such that $\max_{i \in [1, m]} h_i D_i \leq K$ where $D_i$ is the maximum distance between any two consecutive occurrences of task $C_i$. 

\begin{corollary}
The decision version of BGT is NP-hard. 
\end{corollary}

\begin{proof}

Consider the following reduction.

\begin{algorithm}[H]
\caption{Reduction from Pinwheel Scheduling to BGT}\label{algo:bgt}

\KwIn{Pinwheel instance $A = (a_1, \dots, a_m)$}

$L \gets 1$

\For{$i = 1, \dots, m$}{
$L \gets L \cdot a_i$
}

$G \gets ()$

\For{$i = 1, \dots, m$}{
$\begin{aligned}
G \gets G \sqcup \left( \frac{L}{a_i} \right)
\end{aligned}$
}

\Return{$(G, L)$}

\end{algorithm}

Note that the first parameter of the output of \Cref{algo:bgt} should be interpreted as the list of growth rates and the second parameter should be interpreted as the value of $K$.

We claim that the pinwheel instance $A$ is schedulable if and only if the corresponding BGT instance $G$ can be scheduled with objective at most $L$. 

\begin{lemma}
Let $(G, L) = $ \redbgt. If the BGT instance $G$ can be scheduled with objective at most $L$, then the pinwheel instance $A$ is schedulable.
\end{lemma}

\begin{proof}
Recall the BGT objective value as $\underset{i \in [m]}{\max} h_i D_i$ where $D_i$ is the maximum distance between any two consecutive occurrences of job $i$. Note that for any $j \in [n]$, $\underset{i \in [m]}{\max} h_i D_i \geq h_j D_j$. Since the objective is at most $L$, $\underset{i \in [m]}{\max} h_i D_i \leq L$, so for any $j \in [n]$, $h_j D_j \leq L$. By Line 6 of \Cref{algo:bgt}, $h_j = L/a_j$ for all $j \in [n]$, so 
$$D_j \leq \frac{L}{h_j} = \frac{L}{L/a_j} = a_j.$$
But this says that the maximum gap between instances of job $j$ is exactly the maximum pinwheel period of job $j$, showing that the BGT schedule which achieves the objective of at most $L$ is a valid pinwheel schedule for $A$. 
\end{proof}

\begin{lemma}
Let $(G, L) = $ \redbgt. If the pinwheel instance $A$ is schedulable, then the BGT instance $G$ can be scheduled with objective at most $L$.
\end{lemma}

\begin{proof}
Since $A$ is schedulable, there must be a corresponding schedule $S$. We claim that $S$ achieves objective at most $L$ for the BGT instance $G$. In particular, recall the BGT objective value as $\max_{i \in [1, m]} h_i D_i$ where $D_i$ is the maximum distance between any two consecutive occurrences of job $i$. By the pinwheel constraint, $D_i \leq a_i$ for all $i$. In addition, by Line 6 of \Cref{algo:bgt}, for any $j \in [n]$, we have $h_j = L/a_j$. Then, 
$$\max_{i \in [1, m]} h_i D_i \leq \max_{i \in [1, m]} \left( \frac{L}{a_i} \cdot a_i \right) = \max_{i \in [1, m]} L = L,$$ as desired.  
\end{proof}

With this reduction in place, the NP-hardness of pinwheel scheduling (\Cref{thm:pinwheel}) implies the NP-hardness of the decision version of BGT.
\end{proof}

\cite{gkasieniec2024perpetual} note that the version of BGT with $K = H$, with $H$ being the sum of the growth rates, is in NP, due to dense pinwheel scheduling being in NP \citep{holte1989pinwheel}. We fully characterize the complexity of this problem.

\begin{corollary}\label{cor:bgt}
The decision version of BGT with $K = \sum_{h = 1}^n h_i$ is NP-complete. 
\end{corollary}

\begin{proof}
As discussed above, decisional BGT with $K = H$ is in NP due to the reduction to pinwheel scheduling.

Now, notice that in \Cref{algo:bgt}, dense pinwheel instances are reduced to bamboo garden trimming instances satisfying 
$$\sum_{i = 1}^n h_i = \sum_{i = 1}^{n} \frac{L}{a_i} = L \sum_{i = 1}^n \frac{1}{a_i} = LD(A) = L \cdot 1 = L = K.$$

We have proven dense pinwheel scheduling to be NP-hard (\Cref{cor:dps}), so this means bamboo garden trimming with $K = H$ is NP-hard as well. Since the problem is in NP and is NP-hard, this proves NP-completeness. 
\end{proof}

Consider the Windows Scheduling Problem \citep{windowsoriginal}, in which one is given an integer $n$ corresponding to the number of jobs, an integer $h$ corresponding to the number of machines, and $(w_1, \dots, w_n)$, the job periods. The problem is to decide whether the jobs can be scheduled with one per timeslot per machine so that each job appears at least once within any interval of its specified period in at least one machine. The problem has been studied with and without machine migration (whether a single job type is allowed to be scheduled across multiple machines) \citep{bar2007windows}.

\begin{corollary}
Windows Scheduling is NP-hard. 
\end{corollary}

\begin{proof}
Windows scheduling is a generalization of pinwheel scheduling. In particular, with $h = 1$, the problem is exactly the pinwheel problem, so windows scheduling inherits the NP-hardness of pinwheel scheduling (\Cref{thm:pinwheel}).
\end{proof}

Consider the constant gap problem \citep{kubiak2004fair, kubiak2009proportional} in which one is given job demands $(d_1, \dots, d_n)$ and asked whether there are integers $f_1, \dots, f_n$ such that $\Big\{ \left( f_1, \frac{D}{d_1} \right), \dots, \left( f_n, \frac{D}{d_n} \right) \Big\}$ is an exact covering sequence. Note that $D = \sum_{i = 1}^n d_i$ and an exact covering sequence is a sequence such that for any $k \in \mathbb{N}$, there is exactly one integer $i \in [n]$ such that $k \equiv f_i \ \bigl(\bmod\ \tfrac{D}{d_i}\bigr)$. 

\begin{corollary}
The constant gap problem is NP-complete. 
\end{corollary}

\begin{proof}
We claim that the constant gap problem $(d_1, \dots, d_n)$ is satisfiable if and only if the BGT instance $(d_1, \dots, d_n)$ with $K = \sum_{i = 1}^n d_i = D$ is satisfiable. 

In particular, if there is a schedule for the BGT instance, since it is dense (i.e. $K$ is the sum of the growth rates), the schedule must have job $i$ appear exactly $\frac{D}{d_i}$ steps apart, meaning that each $f_i$ can be set to the position of the first occurrence of job $i$ and the result is an exact covering sequence. 

In the other direction, with an exact covering sequence, job $i$ always appears $\frac{D}{d_i}$ steps apart so the BGT objective is $\max_{i \in [n]} d_i \cdot \frac{D}{d_i} = D$, so the exact covering sequence is a satisfying BGT schedule.

The NP-completeness of the constant gap problem follows from that of BGT with $K = H$ (\Cref{cor:bgt}).
\end{proof}

Consider the Recurrent Scheduling problem \citep{recharging}, in which one is given weakly concave and weakly increasing functions $H_i(t)$ representing the payout of arm $i$ if the last pull of arm $i$ was $t$ days ago, and we seek to maximize reward. In the decision version of the problem, one is given $L \in \mathbb{R}^+$ and asked whether there is a schedule with value at least $L$. In particular, the value of a schedule $\sigma$ is 
$$\limsup_{T \to \infty} \frac{1}{T} \sum_{i=1}^{n} v_{i,T}(S_i(\sigma|_T)).$$

\begin{corollary}
The decision version of Recurrent Scheduling is NP-hard.
\end{corollary}

\begin{proof}
Consider the following reduction. 

\begin{algorithm}[H]
\caption{Reduction from Dense PS to Recurrent Scheduling, \redrs}\label{algo:recurrent}

\KwIn{Pinwheel instance $A = (a_1, \dots, a_m)$}

\For{$i = 1, \dots, m$}{
$\begin{aligned}
H_i(t) \gets \min\left( 1, \frac{t}{a_i} \right)
\end{aligned}$
}

$L \gets 1$

\Return{$(H, L)$}

\end{algorithm}

Note that for any $c \in \mathbb{N}$, $f(x) = \min(x, c)$ is concave and increasing, demonstrating that the $H_i$ functions of \Cref{algo:recurrent} satisfy the precondition. We claim that $A$ is schedulable if and only if there is a schedule for the corresponding Recurrent Scheduling instance $H$ with value at least $L$.

\begin{lemma}
Let $(H, L) = $ \redrs. If $A$ is not schedulable, then the Recurrent Scheduling instance $H$ does not admit a schedule with value at least $L$.
\end{lemma}

\begin{proof}
Let us define $P = \prod_{i = 1}^m a_i$, and let us define $b = \max(a_1, \dots, a_m)$. 

Consider any arbitrary schedule $S$ for the Recurrent Scheduling problem.

If $A$ is not schedulable, then $S$ must exhibit a pinwheel constraint violation at least once every $P$ days. If not, then there are $P + 1$ states in the state graph induced by these $P$ days ($P$ states before each day and 1 state after the last day). There are $\prod_{i = 1}^m a_i$ valid states so by the Pigeonhole Principle, two of the $P + 1$ states must be the same. This induces a valid pinwheel schedule, as in \Cref{lem:periodic} \citep{holte1989pinwheel}, contradicting the assumption that $A$ is not schedulable. Now, for every violation, the reward is at most 
$$\frac{a_i - 1}{a_i} = 1 - \frac{1}{a_i} \leq 1 - \frac{1}{b}.$$

Compared to a base rate of 1 reward per day, suffering a loss of at least $\frac{1}{b}$ at least once every $P$ days implies that the total reward of any prefix of length $T$ is at most $T - \lfloor T/P \rfloor/b$, for a per-day reward of at most 
$(T - \lfloor T/P \rfloor/b)/T$ which for $T > P$ satisfies 
$$\frac{T - \lfloor T/P \rfloor/b}{T} \leq \frac{T - (T/2P)/b}{T} = 1 - \frac{1}{2Pb},$$ 
meaning that the relevant limit is at most $1 - \frac{1}{2Pb}$. Since $L = 1 > 1 - \frac{1}{2Pb}$, this completes the proof. 
\end{proof}

\begin{lemma}
Let $(H, L) = $ \redrs. If $A$ is schedulable, then the Recurrent Scheduling instance $H$ admits a schedule with value at least $L$.
\end{lemma}

\begin{proof}
Since $A$ is schedulable, there must be some valid schedule $S$. Note that we have restricted the reduction in \Cref{algo:recurrent} to only reduce from dense pinwheel instances. This means that $S$ is an EPS schedule for $A$ \citep{covering}. In other words, on every day, when job $i$ is scheduled, it has been exactly $a_i$ days since the last job, implying a reward of 1 upon scheduling the job (this does not necessarily hold for the first time a job is scheduled, but that is a boundary effect which does not affect the overall limit). This shows that schedule $S$ achieves value $L = 1$ for the Recurrent Scheduling problem, completing the proof. 
\end{proof}

We can see that \Cref{algo:recurrent} runs in polynomial time, so the NP-hardness of Recurrent Scheduling now follows from the NP-hardness of dense pinwheel scheduling (\Cref{cor:dps}). 
\end{proof}

%% file: ptas-proof.tex
\section{Omitted Proofs from \texorpdfstring{\Cref{sec:ptas}}{Section~\ref{sec:ptas}}}
\label{sec:ptas-proof}

We restate and prove several lemmas whose proofs were deferred to this section. 

\lemmaruntime*

\begin{proof}
We first prove that the while loop on Line 6 of \Cref{algo:pptas} terminates. 

Consider the series $\mathcal{S}$, with $\mathcal{S}_1 = n$ (as defined in Line 1 of \Cref{algo:pptas}) and $\mathcal{S}_{i + 1} = 3n^2 \cdot \mathcal{S}_i^{\mathcal{S}_i}$. Then, on iteration $i$ of the while loop on Line 6 of \Cref{algo:pptas}, the value of $\ell$ is $\mathcal{S}_i$, and the value of $u$ is $\mathcal{S}_{i + 1}$.

Let us define $D(i)$ as the total density of jobs strictly between $\mathcal{S}_i$ and $\mathcal{S}_{i + 1}$. More formally, let $j(i)$ be the list of job periods $p$ in $A$ satisfying $\mathcal{S}_i \leq j < \mathcal{S}_{i + 1}$. Then, $D(i) = \sum_{j \in j(i)} \frac{1}{j}$.

Now, since the $j(i)$ lists are non-overlapping, note that 
$$\sum_{i = 1}^{2n + 2} D(i) = \sum_{i = 1}^{2n + 2} \sum_{j \in j(i)} \frac{1}{j} \leq \sum_{j \in A \text{ s.t. } a_1 \leq p < a_{2n + 2}} \frac{1}{j} \leq \sum_{j \in A} \frac{1}{j} \leq 1.$$

The last inequality follows from the fact that if the total density was greater than 1, we already declared $A$ as unschedulable on Line 3 of \Cref{algo:pptas}. 

Since $\sum_{i = 1}^{2n + 2} D(i) \leq 1$ and $D(i) \geq 0$ for all $i \in [2n + 2]$, there must exist an $i \in [2n + 2]$ such that $D(i) \leq \frac{1}{2n + 2}$. 

Therefore, there are at most $2 \lceil 1/\epsilon \rceil + 2$ iterations of the while loop, which is constant with respect to $m$. The density calculation on Line 8 of \Cref{algo:pptas} takes $O(m)$ time. 

For the while loop on Line 15 of \Cref{algo:pptas}, note that $\ell \leq \mathcal{S}_{2n + 2}$, and since $2n + 2$ is a constant with respect to $m$, $c_1 = \mathcal{S}_{2n + 2}$ and $c_2 = (\mathcal{S}_{2n + 2})^{\mathcal{S}_{2n + 2}}$ are constants with respect to $m$ as well. Since $D(A_{\tbig}) \leq 1$ (otherwise we would have $D(A) > 1$ and termination on Line 3 of \Cref{algo:pptas}) and each element has period at most $c_1$, there are at most $c_1$ elements in $A_{\tbig}$, and at most $c_1 + 1$ possibilities for each day in $s$ (due to allowing for holidays). 

It is possible to check whether a sequence of length $p$ could form an infinite $p$-periodic scheduling (by infinitely copying the sequence) in $O(p)$ time by checking the constraints for each job internal to the schedule (i.e., ensuring that the gap between consecutive occurrences is at most the job period) and then checking the corresponding gap between the first and last scheduling of each job (with $p$ being added to the index of the first scheduling) within the sequence to account for boundary effects. 

Similarly, Lines 16 to 18 of \Cref{algo:pptas} can be implemented in a finite amount of time (i.e., without having to consider the whole infinite schedule $s$) by noting that only a single period is relevant, so considering a valid (as defined above) length $p$ sequence is enough. 

Now, the time required to run the while loop can be upper bounded by $\sum_{p = 1}^{c_2} (c_1 + 1)^p \cdot O(p)$, which is the time needed to brute-force enumerate and evaluate all possible sequences of length $p$ for all relevant $p$. Now, 
$$\sum_{p = 1}^{c_2} (c_1 + 1)^p \cdot O(p) \leq c_2 \cdot (c_1 + 1)^{c_2} \cdot O(c_2) \leq O\left( c_2^{c_2 + 2} \right).$$
Recalling that $c_2$ is a constant with respect to $m$, this shows that the entire while loop can be executed in constant time with respect to $m$. The other components of the algorithm clearly require at most $O(m)$ time, so the overall complexity is $O(m)$. 
\end{proof}

\lemmaunschedulable*

\begin{proof}
First of all, if \Cref{algo:pptas} returns on Line 3, we can clearly see that $A$ is unschedulable due to having density greater than 1. Now, we have to prove that if \Cref{algo:pptas} returns on Line 22, $A$ is unschedulable. 

Suppose, for the sake of contradiction, that $A$ is schedulable. Then, it must have a cyclic schedule of period length at most $\prod_{i = 1}^m a_i$ by \Cref{lem:periodic} \citep{holte1989pinwheel}. Let us denote one period of this schedule as $S_1$.

Let us denote $B_1$ and $B_2$ as the values of $\ell$ and $u$ on Line 12 of \Cref{algo:pptas}, respectively.

Consider the schedule $S_2$ which replaces all jobs scheduled in $S_1$ that are ``not big'' (i.e., with period greater than $B_1$) with holidays. Clearly, the fraction of holidays in $S_2$, which we denote as $h(S_2)$, is at least $D(A_{\text{not big}})$. Note that $D(A_{\text{not big}}$ is greater than $h_{\max}$, by assumption, from Line 21. 

In other words, we have demonstrated that one period of a periodic schedule of $A$ has a fraction of holidays strictly greater than $h_{\max}$. Note that this does not immediately contradict the process on Line 18 of \Cref{algo:pptas} because the period of $S_2$ could be greater than $B_1^{B_1}$ (thereby bypassing Line 15 of \Cref{algo:pptas}). 

Let us define the state for the jobs still in $S_2$ (i.e., the big jobs) before day $i$ of $S_2$ as $s_i$. To create the contradiction, we prove the following statement by strong induction:

For any $p \in \mathbb{N}$, the fraction of holidays in a period of any $p$-periodic schedule of $A_{\tbig}$ is at most $h_{\max}$. 

For the base cases, note that for any $p \leq B_1^{B_1}$, the statement follows directly from the process in Lines 15 to 19 of \Cref{algo:pptas}.

For the inductive case, suppose that $p > B_1^{B_1}$. Consider a period of any $p$-periodic schedule $S_p$. By the Pigeonhole Principle, there must be some state vector that is repeated more than once amongst the states which occur before each of the first $p$ days. Suppose one such repetition occurs between days $i$ and $j$ (i.e., $s_i = s_j$). Then, $j > i > 0$ and the schedule which repeats those elements between $i$ and $j - 1$ (inclusive on both ends) is a period of a $(j - i)$-periodic schedule of $A_{\tbig}$, which we call $S_p'$. 

Since $i \geq 1$, by the inductive hypothesis, the fraction of holidays in $S_p'$ is at most $h_{\max}$. 

In addition, consider $S_p''$, the period of a periodic schedule obtained by deleting the days between $i$ and $j - 1$ inclusive. This is still a valid period of a periodic schedule for $A_{\tbig}$ because deleting a proper subcycle in the state graph cannot break the larger cycle. In other words, since the state before $i$ and after $j - 1$ are the same, deleting everything in between cannot be the cause of a constraint violation. 

As with $S_p'$, the fraction of holidays in $S_p''$ is at most $h_{\max}$. In addition, because each day in $S_p$ is found in exactly one of $S_p'$ or $S_p''$, the fraction of holidays in $S_p$ is a convex combination of $h(S_p')$ and $h(S_p'')$. Since we have that $h(S_p') \leq h_{\max}$ and $h(S_p'') \leq h_{\max}$, this establishes that $h(S_p) \leq h_{\max}$.

This completes the inductive step, so by the principle of mathematical induction, the statement is true for any $p \in \mathbb{N}$.

Recall that $S_2$ is a period of a periodic schedule of $A_{\tbig}$ and that $h(S_2) > h_{\max}$. Instantiating the statement with $p$ being the period of $S_2$ immediately yields the desired contradiction. This completes the proof of unschedulability in the case of a return on Line 22 of \Cref{algo:pptas}.
\end{proof}

\lemmaschedulable*

\begin{proof}
Note that in this proof, in addition to proving correctness of the algorithm, we will aim to provide an efficient construction of an (efficient representation of a) schedule for $A(1 + \epsilon)$. 

We will assume that the input $\epsilon$ is the inverse of an integer $n$. If not, Line 1 of \Cref{algo:pptas} has the effect of replacing $\epsilon$ with $\epsilon' = 1/\lceil 1/\epsilon \rceil$ and $\epsilon' \leq \epsilon$ so by the monotonicity property, the schedulability of $A(1 + \epsilon')$ implies schedulability of $A(1 + \epsilon)$.

Let us denote $B_1$ and $B_2$ as the values of $\ell$ and $u$ on Line 1 of \Cref{algo:pptas}, respectively.

Now, we define the terms big jobs, medium jobs, and small jobs in the following way. 

\begin{enumerate}
    \item A job with period $j$ satisfying $j \leq B_1$ is considered a big job. 
    \item A job with period $j$ satisfying $B_1 < p < B_2$ is considered a medium job. 
    \item A job with period $j$ satisfying $j \geq B_2$ is considered a small job.
\end{enumerate}

In addition, we define $D_{\tbig}$, $D_{\tmed}$, and $D_{\tsma}$ as the total density of the big jobs, medium jobs, and small jobs, respectively. 

Note that, by assumption, $h_{\max} \geq D(A_{\text{not big}})$. Let $S_1$ be one period of the $p$-periodic schedule $s$ corresponding to this $h_{\max}$ (which we have shown in \Cref{lem:runtime} to be possible to produce in constant time with respect to $m$).

Let $A_{\tbig}$ be the subset of jobs that are big. We know that $S_1$, when repeated infinitely, is a valid scheduling of $A_{\tbig}$. Let us define $S_2$ as the sequence obtained from repeating $S$ for $n$ times. Then, $S_2$, when repeated infinitely, is also a valid scheduling of $A_{\tbig}$ because the corresponding infinite schedule is the same for both $S_1$ and $S_2$. 

Let us define the operation of ``inserting a holiday'' at position $p$ in schedule $S$ as copying the first $p$ positions of $S$ into $S'$, appending a holiday, then appending a copy of the last $\text{len}(S) - p$ positions of $S$. Note that this operation adds one day to the length of $S$. 

Let us define $S_3$ as the sequence obtained from applying the holiday insertion operation every $n$ days (i.e., $S_3$ contains days $1$ to $n$ of $S_2$, followed by a holiday, followed by days $n + 1$ to $2n$ of $S_2$, followed by a holiday, and so on, ending with a holiday). 

\begin{lemma}
$S_3$, when repeated infinitely, is a valid scheduling of $A_{\tbig} (1 + \epsilon)$. 
\end{lemma}

\begin{proof}
Suppose, for the sake of contradiction, that there is some job $i$ in $A_{\tbig}$ of period $j_i$ which violates its corresponding constraints in $S_3$.

Note that within $A_{\tbig} (1 + \epsilon)$, the period becomes $j(1 + \epsilon)$ and the relevant constraint is that for any $R$, for any range of length $R$ there are at least $\left\lfloor \frac{R}{j(1 + \epsilon)} \right\rfloor$ occurrences of the corresponding job. 

Since the constraint is violated, there must be some $R_0$ for which there are strictly fewer than $\left\lfloor \frac{R_0}{j(1 + \epsilon)} \right\rfloor$ occurrences of job $i$ within a range of length $R_0$. 

Now, let $x = \left\lfloor \frac{R_0}{n + 1} \right\rfloor$. We can view any range of length $R_0$ having $x$ repetitions of contiguous $n + 1$ length blocks, as well as $R_0 - (n + 1) x$ additional days. Each of these $n + 1$ length blocks contains $n$ schedulings pulled from the original $S_2$, as well as a holiday. There must be at most one holiday among the $R - (n + 1) x$ additional days, so all but at most one of them were in the original. 

Therefore, at least $n \cdot x + (R_0 - (n + 1)x - 1)$ days in $S_3$ were in the original schedule $S_2$ contiguously, and due to $S_2$ being a valid schedule, job $i$ must have been present at least 
$$\left\lfloor \frac{nx + R_0 - 1 - (n + 1)x}{j_i} \right\rfloor = \left\lfloor \frac{R_0 - 1 - x}{j_i} \right\rfloor \text{ times.}$$

All of these occurrences are still there in our modified schedule, as we discussed. 

Since the job of period $j_i$ violated its constraint over a range of length $R_0$, it must be that 
$$\left\lfloor \frac{R_0 - 1 - x}{j_i} \right\rfloor < \left\lfloor \frac{R_0 n}{j_i (n + 1)} \right\rfloor.$$

In the special case of $R_0$ being a multiple of $n + 1$, there are no additional days, so there are exactly $nx$ days in $S_3$ that are from $S_2$, which contain at least $\left\lfloor \frac{nx}{j_i} \right\rfloor$ occurrences of job $i$. The constraint violation of job $i$ implies that $\left\lfloor \frac{nx}{j_i} \right\rfloor < \left\lfloor \frac{R_0 n}{j_i (n + 1)} \right\rfloor$. However, this is not possible since $\frac{nx}{j_i} = \frac{nR_0}{j_i (n + 1)}$. So, it cannot be that $R$ is a multiple of $n + 1$.

By similar logic, we have that 
$$\frac{R_0 - 1 - x}{j_i} = \frac{R_0 - 1 - \left\lfloor \frac{R_0}{n + 1} \right\rfloor}{j_i} \geq \frac{R_0 - 1 - \frac{R_0}{n + 1}}{j_i} = \frac{R_0 \cdot \left( 1 - \frac{1}{n + 1} \right) - 1}{j_i} = \frac{R_0 n}{j_i (n + 1)} - \frac{1}{j_i}.$$

Recalling that
$\left\lfloor \frac{R_0 - 1 - x}{j_i} \right\rfloor < \left\lfloor \frac{R_0 n}{j_i (n + 1)} \right\rfloor$, putting the two statements together we have that
$$\frac{R_0 n}{j_i (n + 1)} - \left\lfloor \frac{R_0 n}{j_i (n + 1)} \right\rfloor < \frac{1}{j_i}.$$

In particular, if the decimal part of $\frac{R_0 n}{n + 1}$ were greater than $\frac{1}{j_i}$, then $\frac{R_0 - 1 - x}{j_i}$ would have a decimal component greater than 0 with the same integer component. Note that
$$\frac{R_0 n}{j_i (n + 1)} = \frac{R_0}{j_i} - \frac{R_0}{j_i (n + 1)},$$

so plugging that into $\frac{R_0 n}{j_i (n + 1)} - \left\lfloor \frac{R_0 n}{j_i (n + 1)} \right\rfloor < \frac{1}{j_i}$, we have
$$\frac{R_0}{j_i} - \frac{R_0}{j_i (n + 1)} - \left\lfloor \frac{R_0 n}{j_i (n + 1)} \right\rfloor < \frac{1}{j_i}.$$

Let us define $r = R_0 - (n + 1)x$ so that $R_0 = r + (n + 1)x$. Plugging this in, we have
$$\frac{r + (n + 1)x}{j_i} - \frac{r + (n + 1)x}{j_i (n + 1)} - \left\lfloor \frac{R_0 n}{j_i (n + 1)} \right\rfloor < \frac{1}{j_i}.$$

Simplifying, 
$$0 \leq \frac{r + nx}{j_i} - \frac{r}{j_i (n + 1)} - \left\lfloor \frac{R_0 n}{j(n + 1)} \right\rfloor < \frac{1}{j_i}$$

Now, note that $r \in [0, n]$ so $0 \leq \frac{r}{j_i (n + 1)} < \frac{1}{j_i}$. Therefore, 
$$0 \leq \frac{r + nx}{j_i} - \left\lfloor \frac{R_0 n}{j_i (n + 1)} \right\rfloor < \frac{2}{j_i}.$$

Now, we can see that $r + nx = j_i \cdot z + 1$ for some non-negative integer $z$. Plugging the definition of $r$ back in, we have 
$$r + nx = nx + R_0 - (n + 1)x = R - x,$$
so $R - x = j_i \cdot z + 1$. Now, we see that $R_0 - x - 1 = j_i \cdot z$, so 
$$\left\lfloor \frac{R_0 - x - 1}{j_i} \right\rfloor = z,$$
but this means that 
$$\frac{R_0 n}{j_i(n + 1)} \leq z + \frac{1}{j_i},$$
and since $z$ is an integer either $j_i > 1$ in which case clearly $\left\lfloor \frac{R_0 n}{j_i (n + 1)} \right\rfloor = z$ or $j_i = 1$ in which case $R_0$ is not a multiple of $n + 1$ so $\frac{R_0 n}{n + 1}$ cannot be an integer and we still have $\left\lfloor \frac{Rn}{j_i (n + 1)} \right\rfloor = z$. 

Since $\frac{R_0 - 1 - x}{j_i}$ and $\frac{R_0 n}{j_i (n + 1)}$ both have integer component $z$, this contradicts the original inequality $\left\lfloor \frac{R_0 - 1 - x}{j_i} \right\rfloor < \left\lfloor \frac{R_0 n}{j_i (n + 1)} \right\rfloor$, completing the proof. 
\end{proof}

With the big jobs now scheduled, we now move on to scheduling the medium jobs within the gaps of $S_3$, beginning with a lemma with a relevant proof strategy.

\begin{lemma}\label{label:realps}
For any pinwheel instance $A$ with real periods and density at most $\frac{1}{2}$, $A$ is schedulable.
\end{lemma} 

\begin{proof}
Let us define $A' = \text{fold}_2(A)$. Then, 
$$D(A') < D(A) + \frac{1}{\theta} = D(A) + \frac{1}{2} \leq \frac{1}{2} + \frac{1}{2} = 1.$$
where we have used the $1/\theta$ upper bound on density gain from \Cref{lem:fold} \citep{kawamura}.

Since the total density of $A'$ is less than 1 and each job has period at most 1, there must be exactly one job in $A'$ (unless $A$ is the empty instance which is clearly schedulable) between 1 and 2. Then, $A'$ is schedulable by simply scheduling job 1 every day. By \Cref{lem:fold} \citep{kawamura}, the schedulability of $A'$ implies the schedulability of $A$.
\end{proof}

We use parallel logic to the lemma above, along with casework, to schedule the medium jobs. 

\textbf{Case 0}: $D_{\tmed} = 0$

In this case, there are no medium jobs to schedule. 

\textbf{Case 1}: $D_{\tmed} \leq \frac{1}{4(n + 1)}$

In this case, we can take the medium jobs (those jobs with periods strictly between $B_1$ and $B_2$) and fold them with $\theta = 4(n + 1)$. Since $D_{\tmed} \leq \frac{1}{4(n + 1)}$, after folding, the density is less than $\frac{1}{4(n + 1)} + \frac{1}{4(n + 1)} = \frac{1}{2(n + 1)}$. So, there is exactly one remaining job, and it has period greater than $2(n + 1)$. It can be scheduled using every second holiday created by the holiday insertion operations (since one gap was created every $n + 1$ days). 

\textbf{Case 2}: $D_{\tmed} > \frac{1}{4(n + 1)}$

In this case, we can take the medium jobs (those jobs with periods strictly between $B_1$ and $B_2$) and fold them with $\theta = 2(n + 1)$. We have $D_{\tmed} \leq \frac{1}{2(n + 1)}$ due to the condition on Line 8 of \Cref{algo:pptas}. Therefore, after folding, the density is less than $\frac{1}{2(n + 1)} + \frac{1}{2(n + 1)} = \frac{1}{(n + 1)}$. So, there is exactly one remaining job, and it has period greater than $n + 1$. It can be scheduled using the holidays created by the holiday insertion operations (since one gap was created every $n + 1$ days). 

Let $S_4$ be the schedule induced by the scheduling of medium jobs as described above (note that there is no need to compute the explicit schedule, an efficient representation as associated with the fold operation \citep{kawamura} is enough). In either case, $S_4$ schedules the small jobs and medium jobs from $A(1 + \epsilon)$ (and in fact the $\epsilon$ slack is not needed for the medium jobs). What remains is to prove that in either case, the small jobs can fit into the holidays left in $S_4$. 

Since the fraction of holidays $h_{\max}$ in our initial schedule $S_1$ is at least $D(A_{\text{not big}})$, $D(A_{\text{not big}})/h_{\max} \leq 1$. We will use this to upper bound $D_{\tsma}/h(S_4)$. 

In case 0, we have 
$$h(S_4) = \frac{h_{\max}}{1 + \epsilon} + \frac{\epsilon}{1 + \epsilon},$$
and $D_{\tsma} \leq D(A_{\text{not big}})$ so 
$$\frac{D_{\tsma}}{h(S_4)} \leq \frac{D(A_{\text{not big}})}{\frac{h_{\max}}{1 + \epsilon} + \frac{\epsilon}{1 + \epsilon}} \leq \frac{D(A_{\text{not big}})}{\frac{h_{\max}}{1 + \epsilon} + \frac{h_{\max} \epsilon}{1 + \epsilon}} = \frac{D(A_{\text{not big}})}{h_{\max}} \cdot \frac{1 + \epsilon}{1 + \epsilon} \leq 1 \cdot 1 = 1.$$

In case 1, we have 
$$h(S_4) = \frac{h_{\max}}{1 + \epsilon} + \frac{\epsilon/2}{1 + \epsilon},$$
and $D_{\tsma} \leq D(A_{\text{not big}})$, so 
$$\frac{D_{\tsma}}{h(S_4)} \leq \frac{D(A_{\text{not big}})}{\frac{h_{\max}}{1 + \epsilon} + \frac{\epsilon/2}{1 + \epsilon}} \leq \frac{D(A_{\text{not big}})}{\frac{h_{\max}}{1 + \epsilon} + \frac{h_{\max} \epsilon/2}{1 + \epsilon}} = \frac{D(A_{\text{not big}})}{h_{\max}} \cdot \frac{1 + \epsilon}{1 + \frac{\epsilon}{2}} \leq 1 + \frac{\epsilon}{2}.$$

In case 2, we have $h(S_4) = \frac{h_{\max}}{1 + \epsilon}$, and 
$$D_{\tsma} = D(A_{\text{not big}}) - D_{\tmed} \leq D(A_{\text{not big}}) - \frac{1}{4(n + 1)} = D(A_{\text{not big}}) - \frac{\epsilon}{4(1 + \epsilon)}.$$
Then, 
$$\frac{D_{\tsma}}{h(S_4)} \leq \frac{ D(A_{\text{not big}}) - \frac{\epsilon}{4(1 + \epsilon)}}{\frac{h_{\max}}{1 + \epsilon}} \leq \frac{D(A_{\text{not big}}) - D(A_{\text{not big}}) \cdot \frac{\epsilon}{4(1 + \epsilon)}}{\frac{h_{\max}}{1 + \epsilon}} = \frac{D(A_{\text{not big}})}{h_{\max}} \cdot \left( 1 + \epsilon - \frac{\epsilon}{4} \right) \leq 1 + \frac{3 \epsilon}{4}.$$

In any case, 
$$\frac{D_{\tsma}}{h(S_4)} \leq 1 + \frac{3 \epsilon}{4}.$$

Let us define $A_{\tsma, 1}$ as the set of small jobs in $A$ (i.e., jobs with period at least $B_2$).

Define $A_{\tsma, 2}$ as $A_{\tsma, 1}(1 + \epsilon)$ (i.e., multiply every job period in $A_{\tsma, 1}$ by $1 + \epsilon$). Then, 
$$D(A_{\tsma, 2}) = \frac{D(A_{\tsma, 1})}{1 + \epsilon} = \frac{D_{\tsma}}{1 + \epsilon}.$$

Now, suppose we let $A_{\tsma, 3}$ be the result of rounding down all the job periods in $A_{\tsma, 2}$ down to the nearest multiple of $2 \, \text{len}(S_3)$. Note that $\text{len}(S_1) \leq B_1^{B_1}$ and $\text{len}(S_3) = (1 + \epsilon) \cdot \text{len}(S_1)/\epsilon$. In addition, all job periods in $A_{\tsma, 2}$ are at least $\frac{16B_1^{B_1}}{\epsilon^2} \cdot \left( 1 + \epsilon \right)$. Therefore, the flooring operation increases the density by at most a multiplicative factor $1 + \frac{2B_1^{B_1}/\epsilon}{16B_1^{B_1}/\epsilon^2} = 1 + \frac{\epsilon}{8}$. In other words, $D(A_{\tsma, 3}) \leq D(A_{\tsma, 2}) \cdot \left( 1 + \frac{\epsilon}{8} \right)$. 

Let us define $A_{\tsma, 4}$ as the result of multiplying every job period in $A_{\tsma, 3}$ by $h(S_4)$. Clearly, $D(A_{\tsma, 4}) = D(A_{\tsma, 3})/h(S_4)$. We claim that every job period in $A_{\tsma, 4}$ is an integer. In particular, the denominator of $h(S_3)$ is $\text{len}(S_3)$ (ignoring simplification) and $h(S_4)$ is either $h(S_3)$ (in case 0), $h(S_3) - \frac{1}{2(n + 1)}$ (in case 1), or $h(S_3) - \frac{1}{n + 1}$ (in case 2). So, in any case, $h(S_4)$ has a denominator that is a factor of $2 \, \text{len}(S_3)$ (since $\text{len}(S_3)$ already contains a factor of $n + 1$), and every element of $A_{\tsma, 3}$ is a multiple of $2 \, \text{len}(S_3)$.  

We claim that $A_{\tsma, 4}$ is schedulable. In particular, recall that $\frac{D_{\tsma}}{h(S_4)} \leq 1 + \frac{3 \epsilon}{4}$ so $D_{\tsma} \leq h(S_4) \cdot \left( 1 + \frac{3 \epsilon}{4} \right)$. Now, following through the inequalities in the paragraphs above, we have
$$D(A_{\tsma, 4}) = \frac{D(A_{\tsma, 3})}{h(S_4)} \leq \frac{D(A_{\tsma, 2}) \cdot \left( 1 + \frac{\epsilon}{8} \right)}{h(S_4)} = \frac{D_{\tsma}}{1 + \epsilon} \cdot \frac{1 + \frac{\epsilon}{8}}{h(S_4)} \leq \frac{\left( 1 + \frac{\epsilon}{8} \right) \left( 1 + \frac{3\epsilon}{4} \right)}{1 + \epsilon} \leq \frac{1 + \frac{7 \epsilon}{8} + \frac{3 \epsilon^2}{32}}{1 + \epsilon}.$$

Recall that $\epsilon < \frac{2}{7}$ and we are assuming that $\epsilon$ is the inverse of an integer so $\epsilon \in (0, \frac{1}{4}]$. Then,
$$\frac{1 + \frac{7 \epsilon}{8} + \frac{3 \epsilon^2}{32}}{1 + \epsilon} \leq \frac{1 + \frac{7 \epsilon}{8} + \frac{3 \epsilon}{128}}{1 + \epsilon} = \frac{1 + \frac{115 \epsilon}{128}}{1 + \epsilon} = 1 - \frac{\frac{13 \epsilon}{128}}{1 + \epsilon} \leq 1 - \frac{\frac{13 \epsilon}{128}}{1 + \frac{1}{4}} = 1 - \frac{13\epsilon}{160}.$$

We now appeal to the following theorem of \cite{covering} to prove that $A_{\tsma, 4}$ can be scheduled.

\begin{lemma}[\citealp{covering}, Theorem 4]
\label{lem:coveringdensity}
If $A$ is a pinwheel instance with smallest period $f_1$ and 
$$D(A) \leq 1 - \frac{1 + \ln(2)}{1 + \sqrt{f_1}} - \frac{3}{2f_1},$$
then $A$ is schedulable. 
\end{lemma}

Relevant ideas in the proof of \Cref{lem:coveringdensity} are a use of the fold operation (\Cref{algo:fold}) and rounding fold outputs to a close multiple of $\sqrt{f_1}$. We refer to \cite{covering} for the details. 

Note that on Line 5 of \Cref{algo:pptas}, $u$ is initially set to $16n^2 \ell^{\ell}$ and since $n \geq 4$, $16n^2 \ell^{\ell} \geq 16n^2 \cdot 4^4 = 4096n^2$. This means that $f_1 \geq 4096n^2$. Now, 
$$1 - \frac{1 + \ln(2)}{1 + \sqrt{4096n^2}} - \frac{3}{2 \cdot 4096n^2} \geq 1 - \frac{2}{64n} - \frac{1}{64n} = 1 - \frac{3}{64n} = 1 - \frac{3 \epsilon}{64}.$$ 
In addition, $13/160 > 3/64$, so 
$$1 - \frac{13 \epsilon}{160} < 1 - \frac{3 \epsilon}{64}$$
Since $D(A_{\tsma, 4}) \leq 1 - \frac{13 \epsilon}{160}$, this means that by \Cref{lem:coveringdensity} \citep{covering}, $A_{\tsma, 4}$ is schedulable. Then, by \Cref{lem:periodic} \citep{holte1989pinwheel}, there is a corresponding cyclic schedule which we call $S_5$. In addition, it is possible to get an efficient representation of this schedule from the technique in \cite{covering}, which relies on the fold operation to construct a satisfying schedule.

We claim that a scheduling for the small jobs can be obtained by placing the jobs from $S_5$ consecutively within the holidays of $S_4$. Call this schedule $S_6$. 

To see why the small jobs all satisfy their constraints under such a scheme, consider any arbitrary job $i$ with period within $A$ of $j_i$. Then, the corresponding job $i$ has period $j_i$ within $A_{\tsma, 1}$, period $j_i \cdot (1 + \epsilon)$ within $A_{\tsma, 2}$, period $2 \, \text{len}(S_3) \cdot \left\lfloor \frac{j_i \cdot (1 + \epsilon)}{2 \, \text{len}(S_3)} \right\rfloor$ within $A_{\tsma, 3}$, and finally period $2h(S_4) \, \text{len}(S_3) \cdot \left\lfloor \frac{j_i \cdot (1 + \epsilon)}{2 \, \text{len}(S_3)} \right\rfloor$ within $A_{\tsma, 4}$. 

Now, note that the pattern of holidays in $S_4$ repeats every $2 \, \text{len}(S_3)$ days due to cases 0, 1, and 2 of medium job scheduling, either taking no holidays, every other inserted holiday, or every inserted holiday. In particular, for any $d > 0$, day $d$ in $S_4$ is a holiday if and only if day $d + 2 \, \text{len}(S_3)$ is a holiday. 

Therefore, holidays that are $h(S_4) \cdot 2 \, \text{len}(S_3)$ holidays apart must be $2 \, \text{len}(S_3)$ days (including holidays) apart. Iterating this, we see that for any $k > 0$, holidays that are $k \cdot h(S_4) \cdot 2 \, \text{len}(S_3)$ holidays apart must be $k \cdot 2 \, \text{len}(S_3)$ days apart. 

Instantiating this with $k = \left\lfloor \frac{j_i \cdot (1 + \epsilon)}{2 \, \text{len}(S_3)} \right\rfloor$, and noting that the constraint for job $i$ within $A_{\tsma, 4}$ corresponds to the number of holidays apart in $S_6$, we see that consecutive occurrences of job $i$ within $S_6$ must be at most $2 \, \text{len}(S_3) \cdot \left\lfloor \frac{j_i \cdot (1 + \epsilon)}{2 \, \text{len}(S_3)} \right\rfloor$ days apart. In addition, 
$$2 \, \text{len}(S_3) \cdot \left\lfloor \frac{j_i \cdot (1 + \epsilon)}{2 \, \text{len}(S_3)} \right\rfloor \leq 2 \, \text{len}(S_3) \cdot \frac{j_i \cdot (1 + \epsilon)}{2 \, \text{len}(S_3)} = j_i \cdot (1 + \epsilon),$$
which is exactly the job period constraint of job $i$ within $A(1 + \epsilon)$, showing that the small jobs all satisfy their period constraints.

Since all jobs are small, medium, or big, and we have scheduled all these corresponding jobs within $A(1 + \epsilon)$ according to their required periods, this demonstrates the correctness of the algorithm in the schedulable case. 

Note that the placements of the big jobs are explicit and that of the medium and small jobs can be considered in an efficient representation (i.e., the proof is constructive) due to relying on the fold operation. 
\end{proof}

%% file: hardness-proof.tex
\section{Omitted Proofs from \texorpdfstring{\Cref{sec:hardness}}{Section~\ref{sec:hardness}}}
\label{sec:hardness-proof}

We restate and prove several lemmas whose proofs were deferred to this section. 

\lemmaepsfill*

\begin{proof}
By assumption, there is an EPS schedule for $A$, so let $S_1$ be such a schedule. Applying \Cref{lem:periodic} \citep{exact}, $S_1$ must be periodic with period $L = \text{LCM}(a_1, \dots, a_m)$, so let $S_2$ be one period of $S$. 

Let $h$ be the fraction of days in $S_2$ where no job is scheduled. Let us define $B' = (hb_1, \dots, hb_n)$. 

We claim that every job period within $B'$ is an integer. In particular, the denominator of $h$ is a factor of $\text{len}(S_2)$, and by assumption every job period in $B$ is a multiple of $L$, which is exactly $\text{len}(S_2)$. 

Let us define $B''$ as the instance derived from adding $\text{LCM}(B') \cdot (1 - D(B'))$ instances of the period $\text{LCM}(B')$ to $B'$. Note that $\text{LCM}(B')$ is the least common multiple of all job periods within $B'$. In addition, the added job periods are integers because $1 - D(B')$ is a fraction composed as a sum and difference of fractions, which are integer multiples of $1/\text{LCM}(B')$, so the overall fraction is an integer multiple of $1/\text{LCM}(B')$. In addition, 
$$D(B'') = D(B') + \text{LCM}(B') \cdot \frac{1 - D(B')}{\text{LCM}(B')} = D(B') + 1 - D(B') = 1.$$
Note that $D(C) \leq 1$, which means that $D(B) \leq h$, and $D(B') = \frac{D(B)}{h}$, so $D(B') \leq 1$. This shows that a non-negative number of jobs are added to go from $B'$ to $B''$.

We claim that $B''$ is schedulable. In particular, with $b'_i = hb_i$, since each job period is multiplied by the same factor, it still holds that $b'_i | b'_j$ for all $1 \leq i < j \leq n$. The new jobs added also satisfy this divisibility relation since the periods are all the LCM of the previous set of jobs. Recall that $D(B'') \leq 1$ (and in fact this holds with equality). The following theorem of \cite{holte1989pinwheel} implies that $B''$ is schedulable. 

\begin{lemma}[\citealp{holte1989pinwheel}, Theorem 3.1]
\label{lem:divschedulable}
If $A = (a_1, \dots, a_m)$ satisfies $a_i | a_j$ for all $1 \leq i < j \leq n$, and $D(A) \leq 1$, then $A$ is schedulable. 
\end{lemma}

\citealp{holte1989pinwheel} proves \Cref{lem:divschedulable} using an explicit fast online scheduler based on a greedy idea. It is also possible to prove this based on a job combination idea related to an algorithm we present in a later proof (\Cref{algo:combine}). 

Let $S_3$ be a schedule of $B''$. Note that because $B''$ is dense, $S_3$ is an EPS schedule \citep{covering}. 

Let $S_4$ be the schedule obtained from deleting all instances of the ``filler'' jobs added to $B''$ (i.e., the ones which were added with period $\text{LCM}(B')$). Then, $S_4$ is a valid EPS schedule for $B'$, since those job placements were not affected. 

We claim that the schedule $S_5$ obtained from placing the jobs from $S_4$ consecutively within the gaps of $S_1$ is an EPS schedule for $C$. In particular, the fact that the jobs in $A$ satisfy their constraints in $S_4$ is clear from the fact that $S_1$ is a valid EPS schedule for $A$. 

Now consider any job $i$ with period $b_i$ from $B$. Note that due to the cyclic nature of $S_1$, for any $j, k \mathbb{N}$, the number of holidays between $j$ and $j + kL$ is exactly $hkL$. Recall that $b_i$ is a multiple of $L$. Letting $j$ be the location of one scheduling of job $i$ and setting $k = b_i/L$, we see that there are $hkL = hb_i$ holidays in between the scheduling at position $j$ of job $i$ and its next deadline. In addition, since $S_4$ is a valid schedule for $B'$, we know that consecutive instances of job $i$ are scheduled exactly $hb_i$ steps apart within $S_3$, aligning with the need to be exactly $hb_i$ holidays apart within $S_1$. 

Since all jobs in $A$ and $B$ satisfy their period constraints without our constructed schedule, this completes the proof. 
\end{proof}

\lemepsredsat*

\begin{proof}
Note that the main ideas of this proof are from \cite{exact}, though that paper does not formally prove the desired statement with all the details. 

Since \redeps \text{} is schedulable, there must be some satisfying schedule, call it $S$. 

We consider subschedules of $S$, in particular, the $2n$ subschedules whose elements share the same congruence class modulo $2n$. For example, one such subschedule consists of all positions in $S$ which are multiples of $2n$. 

Note that every job in \redeps \text{} has a period which is a multiple of $2n$. By the additivity of modular congruence, this means that for any job $j$, every occurrence of job $j$ must occur within a single subschedule. 

Let us consider $x_i = \bar{x}_{i - n}$ for $i \geq n + 1$, for notational convenience.

Now let $j_i$ be one job of period $\trep_1(x_i)$ and similarly let $j_k$ be one job of period $\trep_1(x_k)$. Then, we claim that $j_i$ and $j_k$ occupy the same subschedule if and only if $i = k$. We prove each direction separately. 

If $i \neq k$, then we need to prove that $j_i$ and $j_k$ cannot occupy the same subschedule. Since $\trep_1(x_i)$ is a different prime from $\trep_1(x_k)$, and $\trep_1(x_i)^2$ and $\trep_1(x_k)^2$ are relatively prime, then \Cref{lem:gcd} applies (with $\gcd(a_i, a_j) = 2n$), and the desired claims follows. 

If $i = k$, we need to prove that $j_i$ and $j_k$ appear within the same subschedule. Suppose, for the sake of contradiction, that they appear in separate subschedules, $s_1$ and $s_2$. Then, from the $i \neq k$ case, we already know that any literal other than $x_i$, its corresponding job must reside in a subschedule other than $s_1$ or $s_2$ (and this relationship holds pairwise), meaning there must be at least $2n - 1$ other subschedules. However, this produces at least $2n + 1$ subschedules, greater than our total of $2n$ subschedules. This is a contradiction, completing the proof. 

Based on our work so far, we can associate each of the $2n$ subschedules with a particular literal. 

Now, consider the job $f_i$ for any $i$. If it were placed in any subschedule other than that of $x_i$ or $\bar{x}_i$, \Cref{lem:gcd} would apply (with $\gcd(a_i, a_j) = 2n$) due to $f_i/2n$ not sharing prime factors with any of the other $\trep_2(x_k)$ values. Therefore, $f_i$ must occupy either the same subschedule as $x_i$ or the same subschedule as $\bar{x}_i$. Let us associate $f_i$ occupying the $x_i$ subschedule with setting the variable $x_i$ to false in the 3-SAT formula, and let us associate $f_i$ occupying the $\bar{x}_i$ subschedule with setting the variable $x_i$ to true. 

Then, we claim that this truth assignment, iterated over all the $f_i$ jobs, produces a satisfying assignment for the 3-SAT formula. We now focus on proving the validity of this truth assignment.

Consider any arbitrary clause $C_j$. It has a corresponding clause job of period $\trep_2(C_j)$ which is scheduled in some subschedule. Since the only prime factors of $\trep_2(C_j)/2n$ are $\trep_1(c_{jk})$ for $k \in [3]$, by \Cref{lem:gcd} (with $\gcd(a_i, a_j) = 2n$), the clause job must be scheduled within a subschedule corresponding to one of its constituent literals. Suppose it is literal $x_k$. We claim that $f_k$ is not scheduled within the subschedule of literal $x_k$. Suppose, for the sake of contradiction, that $f_i$ were scheduled within the schedule of literal $x_k$. Suppose $f_k$ occupies the congruence class $c$ modulo $2n \, \trep_1(x_k)$. 

Applying \Cref{lem:gcd} (with $\gcd(a_i, a_j) = 2n$), the only jobs which could reside in the same congruence class modulo $2n \, \trep_1(x_k)$ are the $\trep_2(x_k)$ jobs and the clause jobs which contain $\trep_1(x_k)$ as a factor in the period. However, all of these jobs have a period, which, when divided by $2n \, \trep_1(x_k)$, are relatively prime to $\trep_1(\bar{x}_k)$ (note that we have assumed a clause cannot contain both a literal and its negation). Then, applying \Cref{lem:gcd} (with $\gcd(a_i, a_j) = 2n \, \trep_1(x_k)$), we see that no job can occupy the same congruence class as $f_k$ modulo $2n \, \trep_1(x_k)$. 

However, this means that the remaining density is $\frac{1}{2n} \cdot \left( 1 - \frac{1}{\trep_1(x_k)} \right)$ which is exactly enough density for the jobs of the form $\trep_2(x_k)$ since 
$$\frac{b_{x_k}}{\trep_2(x_k)} = \frac{[\trep_1(x_k)]^2 - \trep_1(x_k)}{2n[\trep_1(x_k)]^2} = \frac{1}{2n} \cdot \left( 1 - \frac{1}{\trep_1(x_k)} \right).$$
This contradicts the assumption that $C_j$ was scheduled within the subschedule of literal $x_k$, proving that $f_k$ must not have been scheduled within that subschedule. 

However, $f_k$ must be scheduled in either the subschedule of $x_k$ or of $x_{(k + n - 1 \text{ mod } 2n) + 1}$, so it is scheduled in that of $x_{(k + n - 1 \text{ mod } 2n) + 1}$. By our description of the truth assignment, this means that literal $x_k$ is assigned to true (more specifically, if $k \leq n$ and $x_k$ is a positive literal, the variable would be assigned to true and if $k > n$ and $x_k$ is a negation, the variable would be assigned to false but the negation would still satisfy the clause). This proves that clause $C_j$ is satisfied, and we could do this for any $j \in [m]$, completing the proof of satisfiability. 
\end{proof}

\lemepssatred*

\begin{proof}
If $C$ is satisfiable, then there is some satisfying assignment; let $T_i$ be the corresponding truth value of variable $x_i$ in the assignment.

Note that the monotonicity and partitioning properties \citep{kawamura} from pinwheel scheduling have analogs in exact pinwheel scheduling. In particular, if $A \sqcup (a)$ has an EPS schedule, then $A \sqcup (b)$ has an EPS schedule for any $b$ such that $a | b$. In addition, for any $q \in \mathbb{N}$, if $A \sqcup (a)$ has an EPS schedule, then $A \sqcup (\underbrace{aq, aq, \dots, aq}_{q \text{ times}})$ has an EPS schedule. For the proof of the monotonicity property, note that job $b$ can be scheduled within the spaces left by the deletion of job $a$ (with the rest left as gaps). For the partitioning property, one can consecutively schedule jobs $1, 2, \dots, q$ within the gaps left by the deletion of job $a$, and since the original schedule had gaps of exactly $a$, the new one will have gaps between consecutive instances of the same job color of $aq$, as desired. 

We will apply a sequence of unschedulability-preserving transformations, prove that the final list of jobs is schedulable, and thereby prove that the original instance is also schedulable. 

In particular, for all $i$, by the monotonicity property for EPS, the job of period $2n \, \trep_1(x_i) \, \trep_1(\bar{x}_i)$ (corresponding to $f_i$) can be replaced with $2n \, \trep_1(x_i)$. In addition, for all $j$, let $x_{C, j}$ be a literal which satisfies clause $C_j$ according to truth assignment $T$. Such a literal must exist since $T$ is a satisfying truth assignment. Then, for all $j$, the job of period $2n [\trep_1(c_{j_1}) \, \trep_1(c_{j_2}) \, \trep_1(c_{j_3})]^2$ (corresponding to $\trep_2(C_j)$ can be replaced with a job of period $2n[\trep_1(x_{C, j})]^2$. 

Now, we assign subschedule $i$ to literal $x_i$ for all $i$, and describe how they can be scheduled separately within their subschedules.

For any literal $x_i$ which has a truth value of false, the jobs which need to be scheduled are the $[\trep_1(\bar{x}_i)]^2 - \trep_1(\bar{x}_i)$ repetitions of period $[\trep_1(x_i)]^2$ and the monotonically reduced period from $f_i$ of $\trep_1(x_i)$. Note that we have divided by $2n$ because we are considering the placement only within the subschedule, and we know they do not conflict with each other. Let $B$ be the corresponding pinwheel instance. Then, 
$$D(B) = \frac{[\trep_1(\bar{x}_i)]^2 - \trep_1(\bar{x}_i)}{[\trep_1(x_i)]^2} + \frac{1}{\trep_1(x_i)} = \frac{\trep_1(\bar{x}_i) - 1}{\trep_1(x_i)} + \frac{1}{\trep_1(x_i)} = 1.$$
In addition, if the job of period $\trep_1(x_i)$ is placed first within $B$, then all job periods are divisible by the previous period. Then, applying \Cref{lem:epsfill} with $A = ()$, there is an EPS schedule for $B$, as desired. 

For any literal $x_i$ which has a truth value of true, the jobs which need to be scheduled are the $[\trep_1(\bar{x}_i)]^2 - \trep_1(\bar{x}_i)$ repetitions of period $[\trep_1(x_i)]^2$ and monotonically reduced periods from the clauses of the form $\trep_1(x_i)^2$. Let $B$ be the corresponding pinwheel instance. Then, 
$$D(B) = \frac{\trep_1(\bar{x}_i) - 1}{\trep_1(x_i)} + \frac{r}{\trep_1(x_i)^2},$$ 
where $r$ is the number of clause jobs that $x_i$ is assigned to satisfy. Clearly, $r \leq m$, and by the definition of $\trep_1$, $\trep_1(x_i) \geq m$ so 
$$\frac{\trep_1(\bar{x}_i) - 1}{\trep_1(x_i)} + \frac{r}{\trep_1(x_i)^2} \leq \frac{\trep_1(\bar{x}_i) - 1}{\trep_1(x_i)} + \frac{m}{m \cdot \trep_1(x_i)} = \frac{\trep_1(\bar{x}_i) - 1}{\trep_1(x_i)} + \frac{1}{\trep_1(x_i)} = 1.$$
Since all the periods are equal, each period in $B$ is divisible by the last. Now, applying \Cref{lem:epsfill} with $A = ()$, there is an EPS schedule for $B$, as desired. Each literal is true or false, so this completes the proof. 
\end{proof}

\lemprimesize*

\begin{proof}
The proof follows easily using the following lemma.

\begin{lemma}\label{lem:count}
For any $v \geq 3$, the number of primes $p$ satisfying $v < p \leq v^3$ is at least $2v$. 
\end{lemma}

\begin{proof}
We have programmatically verified that the statement is true for $v \in \{ 3, 4, 5, 6, 7, 8 \}$. 

We now handle $v \geq 9$, where the following corollary of \cite{counting} is useful.

\begin{lemma}[\citealp{counting}, Corollary 1]
\label{lem:rosser}
For any $x \geq 17$, 
$$\pi(x) > \frac{x}{\ln(x)},$$
where $\pi(x)$ is the number of primes less than or equal to $x$. 
\end{lemma}

\Cref{lem:rosser} can be considered a concrete version of the asymptotic prime number theorem \citep{hadamard1896, valleepoussin1896}. 

Since $9^3 \geq 17$, we can instantiate \Cref{lem:rosser} with $x = v^3$ to get 
$$\pi(v^3) > \frac{v^3}{\ln(v^3)} \geq \frac{v}{3} \cdot v \cdot \frac{v}{\ln(v)}.$$
Now, note that $\frac{v}{3} \geq 3$ for all $v \geq 9$ and $\frac{v}{\ln(v)} \geq 1$ for all $v > 1$, so 
$$\frac{v}{3} \cdot v \cdot \frac{v}{\ln(v)} \geq 3 \cdot v \cdot 1 = 3v.$$
This means that the number of primes $p$ satisfying $p \leq v^3$ is at least $3v$. Clearly, the number of primes satisfying $p \leq v$ is at most $v$, so the number of primes satisfying $v < p \leq v^3$ is at least $3v - v = 2v$, as desired. 
\end{proof}

Since we are considering 3-SAT formulas, $\max(m, n) \geq 3$, so we can apply \Cref{lem:count} with $v = \max(m, n)$. Then, there are at least $2 \max(m, n)$ primes strictly between $\max(m, n)$ and $\max(m, n)^3$, and Line 2 of \Cref{algo:epsreduction} constructs $2n \leq 2 \max(m, n)$ primes, so this proves the desired lemma.  
\end{proof}

\lemperiodsize*

\begin{proof}
The desired results follow if we prove that the amount by which $\tper$ is multiplied on any single iteration of the for loop on Line 3 of \Cref{algo:allowed} is between $v^{28}$ and $v^{84}$. 

We claim that the number of primes multiplied during any iteration is 28. First of all, $2 \cdot 2 = 4$ primes come from $\trep_1(x_{\tnex})^2 \cdot \trep_1(\bar{x}_{\tnex})^2$ in Line 10 of \Cref{algo:allowed}. If variable $x_i$ is in $k$ clauses, then $2 \cdot 3 \cdot k$ primes come from the fact that each clause has 3 literals in 3,4-SAT, and we multiply by the square of a $\trep_1$ on Line 10 of \Cref{algo:allowed}. Then, $2 \cdot 3 \cdot (4 - k)$ primes are added on Line 12 of \Cref{algo:allowed} (note that $k \leq 4$ since we are in 3,4-SAT) so the total is 
$$2 \cdot 2 + 2 \cdot 3 \cdot k + 2 \cdot 3 \cdot (4 - k) = 4 + 2 \cdot 3 \cdot 4 = 4 + 24 = 28,$$
as desired. By \Cref{lem:primesize}, each relevant prime (i.e. $\trep_1(x_i)$ for some $i$) is between $v$ and $v^3$. So, 
$$v^{28} \leq \frac{\tpers_i[j]}{\tpers_i[j - 1]} \leq v^{28 \cdot 3} = v^{84},$$
and a similar result holds on $\tpers_i[0]$, as desired. 
\end{proof}

\lemgreedypol*

\begin{proof}
To prove that \Cref{algo:greedy} runs in polynomial time, it is sufficient to prove that the value of $\tjob$ is bounded by a polynomial in the input size in each iteration of the for loop on Line 3 of \Cref{algo:greedy}. 

In addition, because we set $d' = d - \lfloor \tper \cdot d \rfloor/\tper$, on each iteration of the for loop on Line 3 of \Cref{algo:greedy}, $d' < 1/\tper$. In particular, for the iteration corresponding to $\tpers[i]$, $d < 1/\tpers[i - 1]$, meaning that 
$$\tjob < \frac{\tpers[i]}{\tpers[i - 1]} \leq \max(m, n)^{84}.$$
This bound holds for all but the first iteration. For the first iteration, $\tpers[0] \leq 2n \max(m, n)^{84}$, and $d \leq \frac{1}{n}$, so the relevant bound is $\tjob \leq 2n \max(m, n)^{84}/n \leq 2 \max(m, n)^{84}$. 
Since the input size is clearly at least $\max(m, n)$, this completes the proof. 
\end{proof}

\lemgreedyd*

\begin{proof}
By the process on Lines 6 to 8 of \Cref{algo:greedy}, the value of $d$ after each iteration is $d_{\text{init}} - D(\tall)$ (specifically for the value of $\tall$ after that same iteration). We need to show that $d$ equals 0 after the final iteration. This will follow if we show that $\tper \cdot d$ is an integer, so that $\lfloor \tper \cdot d \rfloor = \tper \cdot d$. By the structure of \Cref{algo:allowed}, $P = \tpers[n - 1]$ is a multiple of $\tpers[i - 1]$ for all $i \in [n]$. Since the density of $\tall$ immediately prior to the last iteration is a sum of fractions, each of which is an integer multiple of $1/P$ so, the total density is also an integer multiple of $1/P$. In addition, by assumption on the input, $d_{\text{init}}$ is an integer multiple of $1/P$. Thus, immediately before the last iteration, $d = d_{\text{init}} - \tall$ is an integer multiple of $1/P$. Then, we can see that $dP$ is an integer, as desired. 
\end{proof}

\lemwarmpol*

\begin{proof}
We have already proven that $\tpers[0]$ is polynomial in the input size and $\tpers[i - 1]/\tpers[i - 2]$ is polynomial in the input size for all $i \in \{ 3, \dots, n \}$ (in \Cref{lem:periodsize}), so the for loop on Line 3 of \Cref{algo:warm} runs in polynomial time. 

For the while loop on Line 8 of \Cref{algo:warm}, it is similar to \Cref{algo:greedy} and in fact, the same proof as \Cref{lem:greedypol} still holds, with the only change being that for the first iteration due to skipping $i = 1$, the bound on $\tjob$ is $\max(m, n)^{84} \cdot 2 \max(m, n)^{84} = 2 \max(m, n)^{168}$, and this is still polynomial in the input size. 
\end{proof}

\lemwarmd*

\begin{proof}
As noted in the proof of \Cref{lem:warmpol}, the while loop on Line 8 of \Cref{algo:warm} is similar to \Cref{algo:greedy}. In fact, the same proof holds for \Cref{algo:warm} as long as we verify the relevant preconditions. 

In particular, note that the process from Lines 5 to 7 of \Cref{algo:warm} preserves the property that after each iteration $d + D(\tall) = d_{\text{init}}$. Additionally, the allowed periods are the same between both for loops so as long as $d_{\text{init}}$ is an integer multiple of $1/P = 1/\tpers[n - 1]$ (which is true according to the assumption on input), $d$ will still be an integer multiple of $1/P$ after the last iteration of the while loop on Line 3 of \Cref{algo:warm}. Finally, we need to prove that the value of $d$ is not negative before the start of the while loop on Line 8 of \Cref{algo:warm}. Since $d_{\text{init}} \geq \frac{\tpers[0]}{n \, \tpers[1]}$, we need to prove that $D(\tall) \leq \frac{\tpers[0]}{n \, \tpers[1]}$ before the start of the while loop on Line 8 of \Cref{algo:warm}. At this point, we have 
$$D(\tall) = \sum_{i = 3}^n \frac{\tpers[0] \cdot \tpers[i - 1]}{2n \tpers[i - 2] \cdot \tpers[i - 1]} = \sum_{i = 1}^{n - 2} \frac{\tpers[0]}{\tpers[i]} = \frac{\tpers[0]}{2n} \cdot \sum_{i = 1}^{n - 2} \frac{1}{\tpers[i]}.$$
Now, note that $\tpers[i] \geq 2 \, \tpers[i - 1]$ for any $i \in [n - 1]$, due to Line 10 of \Cref{algo:allowed}. So, 
$$\sum_{i = 1}^{n - 2} \frac{1}{\tpers[i]} \leq \sum_{i = 0}^{\infty} \frac{1}{\tpers[1] \cdot 2^i} = \frac{1}{\tpers[1]} \cdot \sum_{i = 0}^{\infty} \frac{1}{2^i} = \frac{2}{\tpers[1]}.$$ 
Plugging this back in, we see that $D(\tall) \leq \frac{\tpers[0]}{n \, \tpers[1]}$, as desired. 
\end{proof}

\lempsredsched* 

\begin{proof}
Since \redps \text{} is schedulable, there is some corresponding schedule, $S$. 

By \Cref{lem:density}, \redps \text{} is a dense instance. This means that $S$ is an EPS schedule for \redps \text{} \citep{covering}. Now, if we delete the jobs that were added between Lines 3 and 10 of \Cref{algo:reduction} and their associated schedulings, we are left with an EPS schedule for \redeps, since those jobs are unaffected. Then, by \Cref{lem:epsredsat}, $C$ is satisfiable. 
\end{proof}

\lempsredsat*

\begin{proof}
By \Cref{lem:redconsatisfiable}, \redeps \text{} is schedulable, so let $S$ be a corresponding schedule. 

Applying \Cref{lem:epsperiodic}, $S$ is cyclic with a period that is a multiple of $2n$, and we can associate each congruence class modulo $2n$ with a particular literal. 

For all $i \in [n]$, let $d_{i, 1}$ denote the density of the jobs scheduled in $S_i$, the subschedule of variable $x_i$ in $S$. Recall that the subschedule of a variable is the union of the two subschedules of the corresponding (positive and negative) literals. 

Now, for all $i \in [n]$, let $d_{i, 2}$ be the density of $\text{greedy}(i, d_{\text{remaining}})$ taken from Line 6 of \Cref{algo:reduction}. Note that we will assign such jobs to the corresponding variable subschedule $S_i$. 

For all $i \in [n]$, let $d_{i, 3} = \frac{1}{n} - d_{i, 1} - d_{i, 2}$.

Now, note that if every clause job were scheduled within $S_i$, then we would have $d_{i, 3} = 0$ due to the input to the greedy algorithm on Line 6 of \Cref{algo:reduction}. A visual representation of this is also shown in \Cref{fig:greedy}. 

In addition, each clause job which occupies a subschedule must correspond to a literal which appears in that clause in the 3,4-SAT instance (see the proof of correctness for \Cref{algo:epsreduction} for details). 

Since we are solving 3,4-SAT, this means that in general, all clause jobs will not be scheduled within any single subschedule. The amount of remaining space within subschedule $S_i$ is exactly $d_{i, 3}$, and we want to fill up all the subschedules now with the warm greedy jobs. This induces the density flow graph in \Cref{fig:flow} with $c = \tsum$ and $rs_i = d_{i, 3}$ for $i \in [n]$. 

In particular, the flow capacity from $s$ to $wg_i$ represents the total density of jobs from the warm greedy step on the $i$th iteration of the for loop of \Cref{algo:reduction} being \tsum. A flow of value $f$ from $wg_i$ to $S_i$ represents assigning a set of jobs with density $f$ to be scheduled within subscheduled $S_i$ (we will explain how to do this shortly), and a flow to $S_{i + 1}$ correspondingly means assigning jobs to $S_{i + 1}$. A valid flow that saturates all the edges going out of $s$ would imply being able to assign every warm greedy job to a subschedule. We claim that such a flow exists, and prove it by explicitly constructing such a flow.  

\begin{algorithm}[H]
\caption{Flow Constructor}\label{algo:flow}
\KwIn{$d_{i, 3}$ values for $i \in [n]$}

$d_{1, 4} \gets d_{1, 3}$

$d_{1, 5} \gets \tsum - d_{1, 4}$

\For{$i = 2, 3, \dots, n - 1$}{
$d_{i, 4} \gets d_{i, 3} - d_{i - 1, 5}$

$d_{i, 5} \gets \tsum - d_{i, 4}$
}

\Return{$d_{i, 4}$ \textnormal{and} $d_{i, 5}$ \textnormal{values for} $i \in [n - 1]$}
\end{algorithm}

In particular, the corresponding flow $F$ is the one that saturates all the edges going out of $s$ and all the edges going into $t$. In addition, $F(wg_i, S_i) = d_{i, 4}$ and $F(wg_i, S_{i + 1}) = d_{i, 5}$ for all $i \in [n - 1]$. The $wg_i$ nodes satisfy flow conservation because $d_{i, 4} + d_{i, 5} = \tsum$ for all $i \in [n - 1]$. The $S_i$ nodes satisfy flow conservation because $d_{i, 4} + d_{i - 1, 5} = d_{i, 3}$ for all $i \in [n - 1]$.  Note that this holds on the boundary as well because $\sum_{i = 1}^n d_{i, 3} = \tsum \cdot (n - 1)$. We will prove that capacity constraints are not violated by induction.

\begin{lemma}
\label{lem:flowcapacity}
For any $i \in [n - 1]$, $0 \leq d_{i, 4} \leq \tsum$ and $0 \leq d_{i, 5} \leq \tsum$. 
\end{lemma}

\begin{proof}
For the base case, by the property of the greedy algorithm of leaving density exactly $\tsum$, $0 \leq d_{i, 3} \leq \tsum$ for $i = 1$. 

For the inductive case, $d_{i, 4} + d_{i - 1, 5} = d_{i, 3}$ and by the inductive hypothesis, $0 \leq d_{i - 1, 5} \leq \tsum$, so $0 \leq d_{i, 4} \leq \tsum$. Then, since $d_{i, 5} = \tsum - d_{i, 4}$, the same bounds hold for $d_{i, 5}$ as well. 

By the principle of mathematical induction, this proves that capacity constraints are not violated along any edge by $F$. 
\end{proof}

\Cref{lem:flowcapacity} completes the proof that $F$ is a valid flow. This means the $d_{i, 4}$ and $d_{i, 5}$ values correspond to a valid flow, but this does not specify exactly which of the warm greedy jobs should be assigned along each edge. We start the description of that by introducing a relevant algorithm. 

\begin{algorithm}[H]
\caption{Job Combiner, $\text{combine}(n, \tall, \tpers, \tlev)$}\label{algo:combine}

\KwIn{$\tall$, $\tpers$, $\tlev \in \{ 0, 1 \}$, such that for any $p \in \tall$, we have $p \in \tpers$}

$i \gets n - 1$

\While{$i > \textnormal{\tlev}$}{
$\begin{aligned}
\tfac \gets \frac{\tpers[i]}{\tpers[i - 1]}
\end{aligned}$

$\textsc{repetitions} \gets $ the number of repetitions of the job period $\tpers[i]$ in $\tall$

$\begin{aligned}
\trem \gets \left\lfloor \frac{\textsc{repetitions}}{\tfac} \right\rfloor \cdot \tfac
\end{aligned}$

\For{$j = 1, \dots, \textnormal{\trem}$}{
Remove $\tpers[i]$ from $\tall$
}

\For{$\begin{aligned} j = 1, \dots, \frac{\textnormal{\trem}}{\textnormal{\tfac}} \end{aligned}$}{
\vspace{3pt}
Add $\tpers[i - 1]$ from $\tall$
}

$i \gets i - 1$
}

\Return{$\textnormal{\tall}$}

\end{algorithm}

We now state relevant lemmas about \Cref{algo:combine}.

\begin{lemma}
$D(\text{combine}(n, \tall, \textsc{allowed\_periods}, \tlev)) = D(\tall)$
\end{lemma}

\begin{proof}\label{lem:combinedensity}
This can be proven by induction, specifically, for each iteration of the while loop on Line 2 of \Cref{algo:combine}, the density of the jobs removed is $\trem/\tpers[i]$ and that of the jobs added is $(\trem/\tfac)/\tpers[i - 1]$. Since $\tfac = \tpers[i]/\tpers[i - 1]$,
$$\frac{\frac{\trem}{\tfac}}{\tpers[i - 1]} = \frac{\frac{\trem}{\tpers[i]/\tpers[i - 1]}}{\tpers[i - 1]} = \frac{\trem}{\tpers[i]},$$
as desired. 
\end{proof}

\begin{lemma}\label{lem:combinebound}
Let $A := \text{combine}(n, \tall, \textsc{allowed\_periods}, \tlev)$ and let $r$ be the number of repetitions of the job period $p = \tpers[\tlev]$ in $A$. Suppose that there are no periods smaller than $p$ in $A$. Then, $\frac{r}{p} \leq D(A) < \frac{r + 1}{p}$.
\end{lemma}

\begin{proof}
Let $B$ be the instance that contains all the elements of $A$ except for the $r$ repetitions of the job period $p$. Since there are no periods smaller than $p$ (by assumption), all periods in $B$ must be larger than $p$. In addition, \Cref{algo:combine} preserves the invariant that at all times, the only periods within $\tall$ are also elements of $\tpers$. Then, the job periods in $B$ must be of the form $\tpers[i - 1]$ for $i \in [n]$. Since the number of removed jobs is $\lfloor \textsc{repetitions} / \tfac \rfloor \cdot \tfac$, for each $i$ the maximum number of remaining jobs is less than $\tfac = \tpers[i]/\tpers[i - 1]$. Therefore, the density of $B$ is upper-bounded by 
$$\sum_{i = \tlev + 1}^{n - 1} \frac{\frac{\tpers[i]}{\tpers[i - 1]} - 1}{\tpers[i]} = \sum_{i = \tlev + 1}^{n - 1} \frac{1}{\tpers[i - 1]} - \frac{1}{\tpers[i]}$$
$$= \frac{1}{\tpers[\tlev]} - \frac{1}{\tpers[n - 1]} < \frac{1}{\tpers[\tlev]} = \frac{1}{p}.$$
Clearly, $D(A) = D(B) + \frac{r}{p}$, so the proof that $D(B) < \frac{1}{p}$ shows that $\frac{r}{p} \leq D(A) < \frac{r + 1}{p}$, as desired. 
\end{proof}

We will use an auxiliary job list to facilitate the assignment of jobs into subschedules. In particular, define 
$$G_i = \text{greedy}(i + 1, d_{i, 5})$$
for each $i \in [n - 1]$. Note that this is different from the greedy jobs constructed during \Cref{algo:reduction} because the density parameter is different. Recall that $\tpers_i$ is the output of \Cref{algo:allowed} with $\tind = i$ and the same 3,4-SAT instance $C$. Let us define $wg_i$ as the set of jobs from the warm greedy step on the $i$th iteration of the for loop of \Cref{algo:reduction}. In reality, there will be some sublist of jobs from $wg_i$, call it $wg_{i, 1}$, which end up assigned to subschedule $S_i$, and the remaining sublist of jobs, call it $wg_{i ,2}$, which end up assigned to subschedule $S_{i + 1}$. We won't specify the lists $wg_{i, 1}$ and $wg_{i, 2}$ directly, but rather we will specify $wg'_{i, 1}$ and $wg'_{i, 2}$ which imply schedulability of $wg_{i, 1}$ and $wg_{i, 2}$. 

The following algorithm produces our split. 

\begin{algorithm}[H]
\caption{Job Splitter}\label{algo:splitter}
\KwIn{$\tpers_i$, $wg_i$, $G_i$}

\SetKwProg{Fn}{Function}{}{end}

$wg'_{i, 1} \gets ()$

$wg'_{i, 2} \gets ()$

\Fn{$\textnormal{convert}(m)$}{
    \For{\textnormal{each job period $p$ in $G_i$ for which there are at least $m$ copies of the period $pm$ in $wg_i$}}{
    
    Add $p$ to $wg'_{i, 2}$
    
    Remove $p$ from $G_i$
    
    Remove $m$ copies of $pm$ from $wg_i$
    }
}

$\textnormal{convert}(1)$

$\textnormal{convert}(\tpers_i[0]/2n)$

$\tlev \gets 1$

$wg_i \gets \text{combiner}(n, wg_i, \tpers_{i}, \tlev)$

$\textnormal{convert}(\tpers_i[0]/2n)$

$wg'_{i, 1} \gets wg_i$

\Return{$(wg'_{i, 1}, wg'_{i, 2})$}
\end{algorithm}

We will create a schedulability implication using the following partitioning property of \cite{kawamura}, which shows that for $k \in [2]$, for any pinwheel instance $A$, if $wg'_{i, k} \sqcup A$ is schedulable, then $wg_{i, k} \sqcup A$ is also schedulable. 

\begin{lemma}[\citealp{kawamura}, Lemma 3]
\label{lem:partitioning}
For any pinwheel instance $A$ and any $a, q \in \mathbb{N}$, if $A \sqcup (a)$ is schedulable, then $A \sqcup (\underbrace{aq, aq, \dots, aq}_{q \text{ times}})$ is schedulable. 
\end{lemma}

\begin{proof}
One can delete the schedulings of the job of period $a$ from a schedule of $A \sqcup (a)$ and within the created gaps consecutively schedule the jobs of period $aq$, so that for a single job of that period, there are $q$ gaps between consecutive schedulings, corresponding to at most $aq$ days. 
\end{proof}

The proof of the following lemma will shed light on the actual $wg_{i, 1}$ and $wg_{i, 2}$ lists. 

\begin{lemma}
\label{lem:empty}
$G_i$ is an empty list immediately before \Cref{algo:splitter} terminates.
\end{lemma}

\begin{proof}
Let us define $G_{i, 1}$ as the elements of $G_i$ that originated from the last two iterations of Line 3 of \Cref{algo:greedy}. Let us define $G_{i, 2}$ analogously as the elements that originated from all but the first or last two iterations. And let us define $G_{i, 3}$ as the elements which originated from the first iteration. Note that these lists are a partition of $G_i$, and no jobs are ever added to $G_i$, so we can handle them separately. 

We will show that the $G_{i, 1}$ elements are removed on Lines 8 and 9 of \Cref{algo:splitter}, $G_{i, 2}$ elements are removed on Line 9, and $G_{i, 3}$ elements are removed on Lines 9 and 12. 

Suppose there are $n_1$ jobs of period $P$ in $\tpers_{i + 1}[n - 2]$ and $n_2$ jobs of period $\tpers_{i + 1}[n - 1]$. We know that 
$$n_1 < \frac{\tpers_{i + 1}[n - 2]}{\tpers_{i + 1}[n - 3]} \text{ and } n_2 < \frac{\tpers_{i + 1}[n - 1]}{\tpers_{i + 1}[n - 2]}$$
(see the proof of \Cref{lem:greedypol}). 

Then, we have 
$$n_1 \cdot \frac{\tpers_{i + 1}[n - 1]}{\tpers_{i + 1}[n - 2]} + n_2 \leq \left( \frac{\tpers_{i + 1}[n - 2]}{\tpers_{i + 1}[n - 3]} - 1 \right) \cdot \frac{\tpers_{i + 1}[n - 1]}{\tpers_{i + 1}[n - 2]} + \left( \frac{\tpers_{i + 1}[n - 1]}{\tpers_{i + 1}[n - 2]} - 1 \right)$$
$$\leq \left( \frac{\tpers_{i + 1}[n - 2]}{\tpers_{i + 1}[n - 3]} \right) \cdot \frac{\tpers_{i + 1}[n - 1]}{\tpers_{i + 1}[n - 2]} = \frac{\tpers_{i + 1}[n - 1]}{\tpers_{i + 1}[n - 3]}$$

Applying \Cref{lem:shared} and \Cref{lem:allowedcyclic}, we have that
$$\frac{\tpers_{i + 1}[n - 1]}{\tpers_{i + 1}[n - 3]} = \frac{\tpers_i[n - 1]}{\frac{2n \, \tpers_i[n - 2]}{\tpers_i[0]}} = \frac{\tpers_i[0] \, \tpers_i[n - 1]}{2n \, \tpers_i[n - 2]}$$

Furthermore, on the last iteration of the for loop on Line 3 of \Cref{algo:warm} (which was used to produce $wg_i$), the number of jobs produced with period $\tpers_i[n - 1]$ is 
$$\frac{\tpers_i[0] \, \tpers_i[n - 1]}{2n \, \tpers_i[n - 2]}$$

This was exactly the upper bound on $n_1 \cdot \tpers_{i + 1}[n - 1]/\tpers_{i + 1}[n - 2] + n_2$. Since $\tpers_i[n - 1]$ can be converted to $\tpers_{i + 1}[n - 2]$ at a $(\tpers_{i + 1}[n - 1]/\tpers_{i + 1}[n - 2])$ to 1 rate (on Line 9 of \Cref{algo:splitter}), and $\tpers_i[n - 1]$ can be converted to $\tpers_{i + 1}[n]$ at a 1 to 1 rate (on Line 8 of \Cref{algo:splitter}), this proves that the jobs generated by the for loop on Line 3 of \Cref{algo:warm} within $wg_i$ can be used to delete every element of $G_{i, 1}$. 

The proof for the elements of $G_{i, 2}$ is essentially the same, but without the complication of the boundary condition with $\tpers_{i + 1}[n - 1]$. 

In particular, for any $j \in [n - 3]$, the number of jobs of period $\tpers_i[j + 1]$ that would be needed to eliminate all the occurrences of $\tpers_{i + 1}[j]$ is at most 
$$\frac{\tpers_i[0] \, \tpers_i[j + 1]}{2n \, \tpers_i[j]}$$
This is due to \Cref{lem:greedypol}, \Cref{lem:shared}, and \Cref{lem:allowedcyclic}, as before (for $G_{i, 1}$ we have $j = n - 2$). 

The for loop on Line 3 of \Cref{algo:warm} on iteration $j$ produces 
$$\frac{\tpers_i[0] \cdot \tpers_i[j + 1]}{2n \, \tpers_i[j]}$$
copies of $\tpers_i[j + 1]$. 
So, there are enough occurrences of $\tpers_i[j + 1]$ in $wg_i$ for Line 9 of \Cref{algo:splitter} to eliminate all jobs of the form $\tpers_{i + 1}[j]$ from $G_i$. 

We are left with $G_{i, 3}$, which is composed entirely of jobs that have period $\tpers_{i + 1}[0]$. 

Suppose, for the sake of contradiction, that at least one job of period $\tpers_{i + 1}[0]$ remains in $G_i$ immediately before \Cref{algo:splitter} terminates. 

Note that \Cref{algo:combine} preserves density by \Cref{lem:combinedensity} and \Cref{algo:splitter} also preserves density in the sense that $D(G_i) = D(wg_i)$ after each loop iteration of our algorithm. This is because for a single iteration of the while loop on Line 4, the density removed from $G_i$ is $\frac{1}{p}$ and that from $wg_i$ is $\frac{m}{pm} = \frac{1}{p}$ as well. 

Then, by assumption, immediately before \Cref{algo:splitter} terminates, $D(wg_i) \geq \frac{1}{\tpers_{i + 1}[0]}$. Since at this point $wg_i$ contains no instances of $\tpers_i[0]$ (by construction of \Cref{algo:warm}) and less than $\tpers_i[0]/2n$ repetitions of $\tpers_i[1]$ (otherwise Line 12 of \Cref{algo:splitter} would have continued running), and no jobs are added to $wg_i$ , by \Cref{lem:combinebound},
$$D(wg_i) < \frac{(\tpers_i[0])/2n - 1}{\tpers_i[1]} + \frac{1}{\tpers_i[1]} = \frac{\tpers_i[0])}{2n \, \tpers_i[1]} = \frac{1}{\tpers_{i + 1}[0]}$$
where the last equation comes from applying \Cref{lem:allowedcyclic}. Now, $D(wg_i) < \frac{1}{\tpers_{i + 1}[0]}$ contradicts $D(wg_i) \geq \frac{1}{\tpers_{i + 1}[0]}$, completing the case of $G_{i, 3}$, and with it, the proof. 
\end{proof}

Given \Cref{lem:empty}, we can now see the relationship between $wg_{i, k}$ and $wg'_{i, k}$ for $k \in [2]$. In particular, the $wg_{i, k}$ periods are the corresponding original periods in $wg_i$ before the transformations (i.e., dividing the period by $\tpers_i[0]/2n$, applying the combine algorithm) which resulting in the jobs becoming those of $wg'_{i, k}$. By the contrapositive of \Cref{lem:partitioning} \citep{kawamura} (applied inductively to \Cref{algo:combine}), each of these transformations is unschedulability preserving, so we only need to show schedulability with the job lists $wg'_{i, 1}$ and $wg'_{i, 2}$ to prove schedulability of $wg_{i, 1}$ and $wg_{i, 2}$. 

\begin{lemma}
\label{lem:warmperiods}
Every job in $wg'_{i, 1}$ has a period contained within $\tpers_i$, and every job in $wg'_{i, 2}$ has a period contained within $\tpers_{i + 1}$ 
\end{lemma}

\begin{proof}
For the $wg'_{i, 1}$ part, note that $wg_i$ starts off with job periods that are all contained within $\tpers_i$ and  \Cref{algo:combine} preserves the invariant that the jobs are contained within $\tpers_i$. In addition, no jobs are ever added to $wg_i$ besides Line 11 where \Cref{algo:combine} only produces jobs of period in $\tpers_i$. Therefore, when the remaining jobs are assigned to $wg'_{i, 1}$, the periods are all elements of $\tpers_i$. 

For the $wg'_{i, 2}$ part, note that for any addition to $wg'_{i, 2}$ (on Lines 8, 9, and 12 of \Cref{algo:combine}), the period added is a period contained within $G_i$. Recalling the definition of $G_i$ (specifically that it calls \Cref{algo:warm} with $\tind = i + 1$), by the structure of \Cref{algo:warm}, every job period within $G_i$ is an element of $\tpers_{i + 1}$, as desired. 
\end{proof}

Now, consider all the jobs assigned in the way described to variable subschedule $S_i$. In particular, we have $g_i$, from Line 6 of \Cref{algo:reduction} (different from $G_i$), $wg'_{i - 1, 2}$, and $wg'_{i, 1}$ (we don't have both the $wg'$ values on the boundaries, but the following argument still holds). Let us define $J'_i := g_i \sqcup wg'_{i, 1} \sqcup wg'_{i - 1, 2}$. We claim that every job within $J'_i$ has a period that is contained in $\tpers_i$. This holds for the jobs from $g_i$ since \Cref{algo:warm} only produces jobs with periods in $\tpers_{\tind}$. For the jobs from $wg'_{i, 1}$ and  $wg'_{i - 1, 2}$, the claim holds due to \Cref{lem:warmperiods}. 

For any $i \in [n]$, let $J_i$ be the list of jobs that occupy the subschedule of variable $x_i$ in $S$, the guaranteed schedule from the beginning of this proof (these jobs in $J_i$ are from \redeps). In addition, let us define $J''_i = (J_i/n) \sqcup (J'_i/n)$. Dividing an instance by $n$ denotes dividing every period within that instance by $n$. Note that every job period within $J_i$ and $J'_i$ is a multiple of $2n$ due to the structure of \Cref{algo:epsreduction}, \Cref{algo:greedy}, and \Cref{algo:warm}.

We claim that $J''_i$ is schedulable. In particular, note that $J_i/n$ admits an EPS schedule, because we can simply take the EPS schedule $S$ and remove any day that is not contained within subschedule $S_i$. Recall that every job period in $J'_i$ is an element of $\tpers_i$. If the elements of $J'_i/n$ are sorted in increasing order, then it is true that $(J'_i/n)_j$ divides $(J'_i/n)_k$ for any $1 \leq j \leq k \leq |J'_i|$. This is because the corresponding statement holds for $\tpers_i$, by the multiplicative structure of \Cref{algo:allowed}. In addition, every job period within $J'_i/n$ is a multiple of the LCM of the job periods within $J_i/n$ by the construction of $\tpers_i[0]$ (recall \Cref{eq:lcm}). Because the flow graph of \Cref{fig:flow} was saturated, and job lists of the corresponding density were placed within $wg'_{i, 1}$ and $wg'_{i, 2}$ (as discussed previously), $D(J_i) + (J'_i) = \frac{1}{n}$, so $D(J''_i) = 1$. Finally, this means that \Cref{lem:epsfill} applies with $A = J_i/n$ and $B = J'_i/n$, and we have that $J''_i$ is schedulable. 

Let $S'_i$ be a schedule of $J''_i$, which is guaranteed to exist by the discussion above. Then, we can replace subschedule $S_i$ in $S$ with $S'_i$ by consecutively placing jobs of $S'_i$ within the locations corresponding to $S_i$. In particular, $S_i$ contains 2 days out of every $2n$ days in $S$ (by the definition of variable subschedules), so for a period $p \in$ \redps\, that ends up assigned to subschedule $S_i$, it corresponds to $p/n$ in $J''_i$, and since all such $p$ are multiples of $2n$ (as discussed previously), consecutive instances being apart by exactly $p/n$ in $J''_i$ corresponds to consecutive instances being apart by exactly $(p/n) \cdot (2n/2) = p$ in $S$. This transformation can be applied for all $i \in [n]$ to produce a schedule $S'$ that correctly schedules all jobs from \redeps\,and all jobs assigned to each subschedule. Since every job added to \redps\,between Lines 3 and 10 of \Cref{algo:reduction} is assigned to some subschedule, this completes the proof. 
\end{proof}